\documentclass{CSML}

\def\dOi{11(3:17)2015}
\lmcsheading%
{\dOi}
{1--30}
{}
{}
{Jan.~\phantom04, 2015}
{Sep.~22, 2015}
{}

\ACMCCS{[{\bf Theory of computation}]:  Models of computation---Concurrency---Process calculi}

\usepackage{hyperref}
\usepackage{amssymb}
\usepackage{graphicx}
\usepackage{epsfig}
\usepackage{graphics}
\usepackage{url}
\usepackage{color}

\theoremstyle{plain}

\newenvironment{definition}{\begin{defi} \rm }{\end{defi}}

\newenvironment{proposition}{\begin{prop} \rm }{\end{prop}}
\newenvironment{lemma}{\begin{lem} \rm }{\end{lem}}

\newenvironment{example}{\begin{exa} \rm }{\end{exa}}
{\rm}

\newcommand{\mb}[1]{\mathbb{#1}}
\newcommand{\mc}[1]{\mathcal{#1}}

\newcommand{\nil}{\mathbf{0}}
\newcommand{\mv}[1]{\mathrel{\stackrel{#1}{\rightarrow}}}

\newcommand{\var}[1][]{\ensuremath{\mathit{var}_{#1}}}

\newcommand{\depth}{\ensuremath{\mathit{depth}}}

\newcommand{\myoverline}[1]{\mb{S}(#1)}

\begin{document}

\title[On the Axiomatizability of Impossible Futures]
{On the Axiomatizability of Impossible Futures}

\author[T.~Chen]{Taolue Chen\rsuper a}   
\address{{\lsuper a}Middlesex University London, Department of Computer Science,
The Burroughs, London NW4 4BT, United Kingdom} 
\email{t.chen@mdx.ac.uk}  
\thanks{}   

\author[W.~Fokkink]{Wan Fokkink\rsuper b} 
\address{{\lsuper b}VU University Amsterdam, Department of Computer Science,
De Boelelaan 1081a, 1081 HV Amsterdam, The Netherlands}
\email{w.j.fokkink@vu.nl}  
\thanks{}    

\author[R.~van Glabbeek]{Rob van Glabbeek\rsuper c}
\address{{\lsuper c}NICTA, Locked Bag 6016, Sydney, NSW 1466, Australia,
\and The University of New South Wales, School of Computer Science
and Engineering, Sydney, NSW 2052, Australia}
\email{rvg@cs.stanford.edu}
\thanks{{\lsuper c}NICTA is funded by the Australian Government through the
    Department of Communications and the Australian Research Council
    through the ICT Centre of Excellence Program.}




%

\begin{abstract}
A general method is established to derive a ground-complete axiomatization for a weak semantics from
such an axiomatization for its concrete counterpart, in the context of the process algebra BCCS\@.
This transformation moreover preserves $\omega$-completeness.
It is applicable to semantics at least as coarse as impossible futures semantics.
As an application, ground- and $\omega$-complete axiomatizations are derived for weak failures,
completed trace and trace semantics. We then present a finite, sound, ground-complete axiomatization
for the concrete impossible futures \emph{preorder},
which implies a finite, sound, ground-complete axiomatization for
the \emph{weak} impossible futures preorder. In contrast,
we prove that no finite, sound axiomatization for BCCS modulo concrete and weak
impossible futures \emph{equivalence} is ground-complete.
If the alphabet of actions is infinite, then the aforementioned ground-complete
axiomatizations are shown to be $\omega$-complete. If the
alphabet is finite, we prove that the inequational theories of BCCS
modulo the concrete and weak impossible futures preorder lack such a finite basis.
\end{abstract}

\maketitle


\section{Introduction}
Labeled transition systems constitute a fundamental model of
concurrent computation. Processes are captured by explicitly describing
their states and the transitions from state to state
together with the actions that produce these transitions. A wide range
of notions of behavioral semantics have been proposed, with the aim
to identify those states that afford the same observations.
Notably, van Glabbeek \cite{Gla01} presented the ``\emph{linear
time -- branching time spectrum} I" of behavioral semantics for
finitely branching, \emph{concrete},\footnote{\emph{Concrete}
  processes do not feature the hidden action $\tau$. \emph{Abstract}
  processes, with $\tau$, come with \emph{strong} semantics,
  treating $\tau$ just as any other action, or \emph{weak} ones,
  allowing a degree of abstraction from $\tau$ steps.}
sequential processes. These
semantics are based on simulation notions or on decorated traces.
Fig.~\ref{fig:spectrum} depicts this spectrum,\footnote{Impossible futures semantics was
missing in the original spectrum I \cite{Gla01}, because it was studied seriously only from 2001
on, the year that \cite{Gla01} appeared.}
where an arrow from one semantics to another means that
the target of the arrow is coarser, i.e.\ less discriminating, than
the source. In \cite{Gla93}, 155 weak semantics, which take into
account the hidden action $\tau$, are surveyed. They constitute
the ``\emph{linear time -- branching time spectrum} II" for
finitely branching, \emph{abstract}, sequential processes.

\begin{figure}
\begin{center}
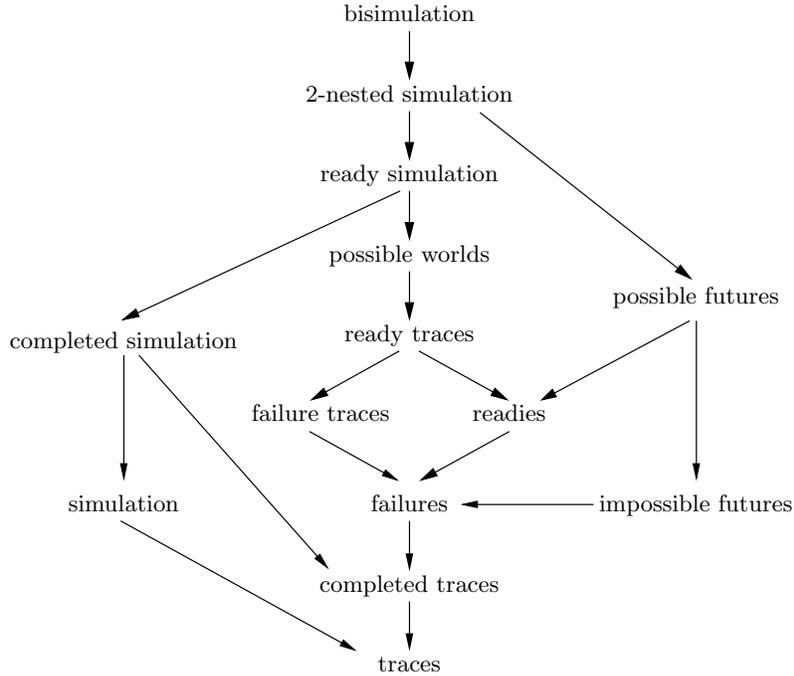
\end{center}
\caption{Linear time-branching time spectrum} \label{fig:spectrum}
\end{figure}

In this paper, we mainly study impossible futures semantics \cite{Vog92,VM01}, which
is a natural variant of possible futures semantics \cite{RB81}.
It is also related to fair testing semantics \cite{RV07}.
Weak impossible futures equivalence is the
coarsest congruence with respect to choice and parallel
composition that contains weak bisimilarity with explicit
divergence, respects deadlock/livelock traces, and assigns
unique solutions to recursive equations \cite{GV06}.

The process algebra BCCSP plays a fundamental role in the study of
concrete semantics. It contains only basic process
algebraic operators from CCS and CSP, but is sufficiently powerful
to express all finite synchronization trees (without
$\tau$-transitions). Van Glabbeek \cite{Gla01} associated with
most behavioral equivalences in his spectrum a \emph{sound}
axiomatization, to equate closed BCCSP terms that are behaviorally
equivalent. These axiomatizations were shown to be
\emph{ground-complete}, meaning that all behaviorally equivalent closed BCCSP
terms can be equated. The process algebra BCCS (see, e.g., \cite{Gla97}) extends BCCSP with
$\tau$, playing the same role as BCCSP in the research of weak semantics.

An \emph{$\omega$-complete} axiomatization enjoys the property that if
all closed instances of an
equation can be derived from it, then the equation itself can
be derived from it. In universal algebra, such an axiomatization
is referred to as a \emph{basis} for the equational theory of the
algebra it axiomatizes. Groote \cite{Gro90} developed a technique
of ``inverted substitutions'' to prove that an axiomatization is
$\omega$-complete, and proved for some equivalences in the
``linear time -- branching time spectrum I" that their equational
theory in BCCSP has a finite basis.

In \cite{AFIL05,CFLN08}, for each preorder and equivalence in
the ``linear time -- branching time spectrum I"
it was determined whether a finite, sound, ground-complete axiomatization
exists. And if so, whether a finite basis exists for the (in)equational theory.
However, for concrete impossible futures semantics the
(in)equational theory remained unexplored.

With regard to the axiomatizability of weak semantics, relatively little
is known compared to concrete semantics.
For some semantics in the ``linear time --
branching time spectrum II" \cite{Gla93}, a sound and
ground-complete axiomatization has been given, in the
setting of BCCS (see, e.g., \cite{Gla97}). Moreover, a finite basis
has been given for weak, delay, $\eta$- and
branching bisimulation semantics \cite{Mi89a,vG93a}.
The inequational theory of BCCS modulo the weak impossible futures preorder
was studied in \cite{VM01}, which offers a finite, sound, ground-complete axiomatization.
Voorhoeve and Mauw also proved that their axiomatization is $\omega$-complete.
It is worth noting that an infinite alphabet of actions is assumed implicitly
\cite[p.\ 7]{VM01}, because a different action is required for each variable.

The current paper studies the axiomatizability of BCCSP and BCCS for semantics
at least as coarse as impossible futures semantics.
In summary, we obtain the following results.

\begin{enumerate}
  \item A link is established between the axiomatizability of concrete and weak semantics.
For any semantics at least as coarse as impossible futures
semantics, an algorithm is provided to turn a ground-complete axiomatization of the
\emph{concrete} version into a ground-complete axiomatization of the corresponding
\emph{weak} version. Moreover, if the former axiomatization is $\omega$-complete, then so is the latter.

\medskip

\item As an application of this algorithm, we derive finite, sound, ground- and
$\omega$-complete axiomatizations for the weak \emph{trace}, \emph{completed trace}
and \emph{failures} preorders and equivalences.
Failures semantics plays a prominent role in the process algebra
CSP \cite{BHR84}. For convergent processes, it coincides with
testing semantics \cite{DH84,RV07}, and thus is the coarsest
congruence for the CCS parallel composition that respects deadlock
behavior. A ground-complete axiomatization for weak
failures equivalence was already given in \cite{Gla97}.

\medskip

\item We provide a finite, sound, ground-complete axiomatization for
BCCSP modulo the concrete impossible futures \emph{preorder} $\precsim_{\rm IF}$.\footnote{In case of an infinite
alphabet of actions, occurrences of action names in axioms should
be interpreted as variables, as else most of the axiomatizations would be infinite.}
(By contrast, no such axiomatization exists for the \emph{possible} futures
preorder \cite{AFGI04}.) Using (1), a finite, sound, ground-complete
axiomatization for the \emph{weak} impossible futures preorder $\precsim_{\rm WIF}$ is obtained.

\medskip

\item We prove that BCCS modulo weak impossible futures \emph{equivalence} $\simeq_{\rm WIF}$
does not have a finite, sound, ground-complete axiomatization. Likewise, we prove that
BCCSP modulo $\simeq_{\rm IF}$ does not have a finite, sound, ground-complete axiomatization.
The infinite families of equations that we use to prove these negative results
are also sound modulo (weak resp.\ concrete) 2-nested simulation equivalence.
Therefore these negative results apply to all BCCS- and BCCSP-congruences that are at least
as fine as impossible futures equivalence and at least as
coarse as 2-nested simulation equivalence. This infers some
results from \cite{AFGI04}, where among others concrete 2-nested simulation
and possible futures equivalence were considered.

\medskip

\item We investigate $\omega$-completeness for impossible futures semantics.
First, we prove that if the alphabet of actions is \emph{infinite}, then the aforementioned ground-complete
axiomatization for BCCS modulo $\precsim_{\rm IF}$ is $\omega$-complete.
Here we apply the technique of inverted substitutions from \cite{Gro90}. Only,
that technique was originally developed for equivalences. Therefore, as an aside, we
adapt this technique in such a way that it applies to preorders.
By (1), this result carries over to $\precsim_{\rm WIF}$.
Second, we prove that in case of a \emph{finite} alphabet of actions, the inequational theories
of BCCS modulo $\precsim_{\rm WIF}$ and of BCCSP modulo $\precsim_{\rm IF}$ do not have a finite basis.
\end{enumerate}

\medskip

\noindent To achieve the negative results, we employ what in
\cite[Sect.~2.3]{AFIL05} is called the proof-theoretic technique.
That is, to prove the nonexistence of a finite axiomatization for
an equivalence $\equiv$ (resp.\ preorder $\sqsubseteq$), it
suffices to provide an infinite family of (in)equations
$t_n \approx u_n$ (resp.\ $t_n\preccurlyeq u_n$) ($n=1,2,3,\dots$)
that are all sound modulo $\equiv$ (resp.\ $\sqsubseteq$), and to
associate with every finite set of sound (in)equations $E$ a
property $P_E$ that holds for all (in)equations derivable from $E$, but does
not hold for at least one of the equations $t_n\approx u_n$ (resp.\ inequations
$t_n\preccurlyeq u_n$). This implies that $t_n\approx u_n$
(resp.\ $t_n\preccurlyeq u_n$) is not derivable from $E$. It follows
that every finite set of sound (in)equations is necessarily incomplete,
so $\equiv$ (resp.\ $\sqsubseteq$) lacks a finite axiomatization.
On top of this, a saturation principle is introduced, to transform
a single summand into a large collection of semi-saturated summands.

Impossible futures semantics is the first example that, in
the context of BCCSP/BCCS, affords a ground-complete axiomatization
modulo the preorder, while missing a ground-complete axiomatization
modulo the equivalence. This surprising
fact suggests that if one wants to show $p \simeq_{\rm IF} q$,
in general one has to resort to deriving $p
\precsim_{\rm IF} q$ and $q \precsim_{\rm IF} p$ separately,
instead of proving it directly. In \cite{AFI07,FG07} an algorithm
is presented which produces from an axiomatization for BCCSP
modulo a preorder, an axiomatization for BCCSP modulo the
associated equivalence. If the original axiomatization for the
preorder is ground-complete or $\omega$-complete, then so is the
resulting axiomatization for the equivalence. In \cite{CFG08}, we
have shown that the same algorithm applies equally well to weak
semantics. However, these algorithms apply only to semantics that
are at least as coarse as ready simulation semantics. Since
impossible futures semantics is incomparable to ready simulation
semantics, it falls outside the scope of \cite{AFI07,FG07,CFG08}.
Interestingly, our results yield that no such algorithm can exist for
impossible futures semantics.

The paper is structured as follows.
Sect.~\ref{sec2} presents basic definitions regarding
the studied semantics, the process algebra's BCCSP and BCCS, and (in)equational logic.
Sect.~\ref{sec:link} describes a transformation of an axiomatization
for a concrete to an axiomatization for a corresponding weak semantics.
In Sect.~\ref{sec:apps} this transformation is applied to
failures, completed trace and trace semantics.
Sect.~\ref{sec:ground} provides finite, sound, ground-complete
axiomatizations for $\precsim_{\rm IF}$ and $\precsim_{\rm WIF}$;
it also presents the aforementioned negative result for $\simeq_{\rm
IF}$ and $\simeq_{\rm WIF}$. Sect.~\ref{sec:omega} is devoted to
the proofs of the positive and negative results regarding
$\omega$-completeness. Sect.~\ref{sec:con} concludes the paper.

This paper combines and extends two previous papers in conference proceedings \cite{CF08} and
\cite{CFG09}. In particular, \cite{CF08} dealt with the concrete impossible futures semantics and \cite{CFG09} extended
it to weak impossible futures and weak failures semantics. Here, new results are presented in Sect.~\ref{sec:link},
which yield a much simplified proof for the results regarding weak failures semantics
given in \cite{CFG09}, and a unified treatment of concrete and weak impossible futures semantics.

\section{Preliminaries} \label{sec2}

\subsection{Labeled transition systems}

Let $A$ be a nonempty, countable set of \emph{concrete} (a.k.a.\ \emph{observable}, \emph{external},
\emph{visible}) actions, which is ranged over by $a,b,c$. The distinguished symbol $\tau$
denotes a \emph{hidden} (a.k.a. \emph{unobservable}, \emph{internal}, \emph{invisible}) action. We assume that
$\tau\notin A$ and write $A_\tau$ for $A\cup\{\tau\}$, which is ranged over by $\alpha$.

\begin{definition}
A \emph{labeled transition system} (LTS) consists of a set of \emph{states} $S$, with typical element $s$, and
a \emph{transition relation} ${\rightarrow}\subseteq S\times A_\tau \times S$.
\end{definition}
We introduce some notation:
$s\mv{\alpha}s'$ means $(s,\alpha,s')$ is an element of $\rightarrow$;
by $s\mv{\alpha}$ we denote that $s\mv{\alpha}s'$ for some $s'$,
and $s\not\mv{\alpha}$ is the negation of this property;
for $\alpha_1\cdots \alpha_k$ a sequence of labels,
$s\mv{\alpha_1\cdots \alpha_k}s'$ means there exist states
$s_0,\dots,s_k$ such that $s=s_0\mv{\alpha_1}\cdots\mv{\alpha_k}s_k=s'$;
the empty sequence is denoted by $\varepsilon$;
we define $\mc{I}(s)=\{\alpha\in A_\tau\mid s\mv{\alpha}\}$;
we write $\Rightarrow$ for the transitive-reflexive closure of $\mv{\tau}$, i.e.,
$\Rightarrow\,\stackrel{\mathrm{def}}{=}(\mv{\tau})^*$;
and $s\Rightarrow\mv{\alpha}$ means there are two states
$s',s''$ with $s\Rightarrow s' \mv{\alpha} s''$.

\subsection{Decorated trace semantics}

\begin{definition}\label{def:traces}
~
\begin{itemize}
\item A sequence $a_1\cdots a_k \in A^*$ with $k\geq 0$ is a
  \emph{trace} of a state $s$ if there is a state $s'$ such that
  $s \mv{a_1\cdots a_k} s'$. It is a \emph{completed trace} of $s$ if moreover
  $\mc{I}(s')=\emptyset$. We write $\mc{T}(s)$ (resp.\ $\mc{CT}(s)$) for the set of
  traces (resp.\ completed traces) of state $s$.
We write $s_1\precsim_{\rm T}s_2$ (resp.\ $s_1\precsim_{\rm CT}s_2$)
if $\mc{T}(s_1)\subseteq \mc{T}(s_2)$ (resp.\ $\mc{CT}(s_1)\subseteq \mc{CT}(s_2)$).\footnote{In \cite{Gla01},
$s_1\precsim_{\rm CT} s_2$ is defined to hold iff $\mc{CT}(s_1)\subseteq \mc{CT}(s_2)$
   \emph{and} $\mc{T}(s_1)\subseteq \mc{T}(s_2)$. Here we can skip the latter condition,
   as we will work with finite transition systems, where $\mc{CT}(s_1)\subseteq \mc{CT}(s_2)$
   implies $\mc{T}(s_1)\subseteq \mc{T}(s_2)$.}

\item A sequence $a_1\cdots a_k \in A^*$ with $k\geq 0$ is a \emph{weak trace} of
  a state $s$ if there is a state $s'$ such that $s\Rightarrow \mv{a_1} \Rightarrow
  \cdots \Rightarrow \mv{a_k} \Rightarrow s'$.
  It is a \emph{weak completed trace} of $s$ if moreover $\mc{I}(s')=\emptyset$.
  We write $\mc{WT}(s)$ (resp.\ $\mc{WCT}(s)$) for the set of
  weak traces (resp.\ weak completed traces) of $s$.
We write $s_1\precsim_{\rm WT}s_2$ if $\mc{WT}(s_1)\subseteq \mc{WT}(s_2)$.
We write $s_1\precsim_{\rm WCT} s_2$ if $\mc{WCT}(s_1)\subseteq \mc{WCT}(s_2)$
and $s_1\mv{\tau}$ implies that $s_2\mv{\tau}$.
  \end{itemize}

\end{definition}

\noindent
The extra requirement that $s_1\mv{\tau}$ implies $s_2\mv{\tau}$, in the definition of $\precsim_{\rm WCT}$,
is needed to make it a precongruence for the process algebra BCCS($A$) (see Remark~\ref{rem:precongruence}).

\begin{definition}\label{def:failures}
~
   \begin{itemize}
\item A pair $(a_1\cdots a_k,B)$ with $k\geq 0$, $a_1\cdots a_k \in A^*$ and $B \subseteq A$
  is a \emph{failure pair} of a state $s$ if there is a state $s'$ such that $s\mv{a_1 \cdots a_k} s'$
  and $\mc{I}(s') \subseteq A\setminus B$. We write $s_1\precsim_{\rm F}s_2$ if the failure pairs of $s_1$ are
  also failure pairs of $s_2$.

  \item A pair $(a_1\cdots a_k,B)$ with $k\geq 0$, $a_1\cdots a_k \in A^*$ and $B \subseteq A$
  is a \emph{weak failure pair} of a state $s$ if there is a state $s'$ such that
  $s\Rightarrow \mv{a_1} \Rightarrow \cdots \Rightarrow \mv{a_k} \Rightarrow s'$
  and $\mc{I}(s') \subseteq A\setminus B$.
  We write $s_1\precsim_{\rm WF} s_2$ if the weak failure pairs of $s_1$ are
  also weak failure pairs of $s_2$ and $s_1\mv{\tau}$ implies that $s_2\mv{\tau}$.
\end{itemize}
\end{definition}

\begin{definition}\label{def:impossible-futures}
~
\begin{itemize}
\item
A pair $(a_1\cdots a_k, B)$ with $k\geq 0$, $a_1\cdots a_k \in A^*$ and $B\subseteq A^*$
is an \emph{impossible future} of a state $s$ if $s\mv{a_1\cdots a_k} s'$ for some state $s'$ with
$\mc{T}(s')\cap B=\emptyset$. We write $s_1\precsim_{\rm IF} s_2$ if the impossible futures of
$s_1$ are included in those of $s_2$.

\item
A pair $(a_1\cdots a_k,B)$ with $k\geq 0$, $a_1\cdots a_k \in A^*$ and $B \subseteq A^*$
  is a \emph{weak impossible future} of a state $s$ if there is a
  trace $s\Rightarrow \mv{a_1} \Rightarrow \cdots \Rightarrow \mv{a_k} \Rightarrow s'$
  with $\mc{WT}(s') \cap B = \emptyset$.
We write $s_1\precsim_{\rm WIF} s_2$ if the weak impossible
futures of $s_1$ are also weak impossible futures of $s_2$,
$\mc{WT}(s_1)=\mc{WT}(s_2)$, and
$s_1\mv{\tau}$ implies that $s_2\mv{\tau}$.
\end{itemize}
\end{definition}

\noindent
The extra requirement that $\mc{WT}(s_1)=\mc{WT}(s_2)$, in the definition of $\precsim_{\rm WIF}$,
is again needed to make it a precongruence for BCCS($A$) (see Remark~\ref{rem:precongruence}).

Given a preorder $\precsim_R$, the associated equivalence is denoted with $\simeq_R$,
where $s_1\simeq_R s_2$ if both $s_1\precsim_R s_2$ and $s_2\precsim_R s_1$.

\subsection{Process algebras BCCS and BCCSP}

BCCS($A$) is a basic process algebra for expressing finite
process behaviors. Its signature consists of the
constant $\nil$, the binary operator $\_+\_$, and unary prefix
operators $\alpha\_$, where $\alpha$ ranges over $A_\tau$.
The process algebra BCCSP($A$) is obtained by excluding the prefix operator $\tau\_$.
In the context of process algebra, $A$ is called the \emph{alphabet}.
Again it is required that $A\neq\emptyset$.

Intuitively, closed BCCS($A$) terms, which are ranged over
by $p,q,r$, represent finite process behaviors, where $\nil$ does
not exhibit any behavior, $p+q$ offers a choice between the
behaviors of $p$ and $q$, and $\alpha p$ executes action $\alpha$ to
transform into $p$.  This intuition is captured by the transition rules below.
They give rise to $A_\tau$-labeled transitions between closed BCCS($A$) terms.
We assume a countably infinite set $V$ of variables;
$x,y,z$ denote elements of $V$, ranging over BCCS($A$) terms.
\[
  \frac{~}{\alpha x\mv{\alpha}x}
\qquad
  \frac{x\mv{\alpha}x'}{x+y\mv \alpha x'}
\qquad
  \frac{y\mv{\alpha}y'}{x+y\mv \alpha y'}
\]

Open BCCS terms, denoted by $t,u,v,w$, may contain variables from $V$.
We write $\var(t)$ for the set of variables occurring in $t$. It is
technically convenient to extend the operational semantics to open
terms. There are no additional rules for variables, which
effectively means that they do not exhibit any behavior.

An occurrence of an action or variable in a term is said to be \emph{initial}
if it is not in the context of a prefix operator.

The \emph{depth} of a term $t$ is the length of the \emph{longest} trace of $t$.
It is defined inductively as follows: $\depth(\nil)=\depth(x)=0$;
$\depth(\alpha t)=1+\depth(t)$; and $\depth(t+u)=\max\{\depth(t), \depth(u)\}$.
The \emph{weak depth} $\depth_w(t)$ does not count $\tau$-transitions,
meaning that it is defined similar to the depth, except $\depth_w(\tau t) = \depth_w(t)$.

A (closed) substitution, denoted by $\rho,\sigma$, maps variables
in $V$ to (closed) terms. Clearly, $t\mv{\alpha}t'$ implies that
$\sigma(t)\mv{\alpha}\sigma(t')$ for all substitutions $\sigma$.
For open terms $t$ and $u$, and a preorder $\sqsubseteq$ (or
equivalence $\equiv$) on closed terms, we define $t\sqsubseteq u$
(or $t\equiv u$) if $\rho(t)\sqsubseteq\rho(u)$ (resp.\
$\rho(t)\equiv\rho(u)$) for all closed substitutions $\rho$.

Summation $\sum \{t_1,\dots,t_n\}$ or
$\sum_{i\in\{1,\ldots,n\}}t_i$ denotes $t_1+\cdots+t_n$, where
summation over the empty set denotes $\nil$. As binding
convention, $\_+\_$ and summation bind weaker than $\alpha\_$. For
every term $t$ there exists a finite set $\{\alpha_it_i\mid i\in
I\}$ of terms and a finite set $Y$ of variables such that
$t=\sum_{i\in I} \alpha_i t_i + \sum_{y\in Y}y$. The $\alpha_i
t_i$ for $i\in I$ and the $y\in Y$ are called the \emph{summands} of $t$.
It is easy to see that $t\mv{\alpha}t'$ iff $\alpha t'$ is a summand of $t$.
The term $\alpha^nt$ is obtained from $t$ by prefixing it $n$
times with $\alpha$, i.e., $\alpha^0t=t$ and
$\alpha^{n+1}t=\alpha(\alpha^nt)$.

A preorder (or equivalence) $R$ is a \emph{precongruence} (resp.\ \emph{congruence})
for BCCS($A$) if $p_1\,R\,q_1$ and $p_2\,R\,q_2$ implies $p_1+p_2\,R\,q_1+q_2$
and $\alpha p_1\,R\,\alpha q_1$ for all $\alpha\in A_\tau$.
If a preorder is a precongruence, then clearly the associated equivalence is a congruence.
The preorders defined in Def.~\ref{def:traces}, \ref{def:failures} and
\ref{def:impossible-futures} are all precongruences for BCCS($A$) \cite{Gla01,Gla93,VM01}.

\begin{rem}
   \label{rem:precongruence}
   The requirement that $s_1\mv{\tau}$ implies $s_2\mv{\tau}$ is used to make $\precsim_{\rm WCT}$,
   $\precsim_{\rm WF}$ and $\precsim_{\rm WIF}$ a precongruence for BCCS($A$).
   Without this requirement we would e.g.\ have $\tau\nil \precsim_{\rm WIF} \nil$.
   However, $\tau\nil+a\nil \not\precsim_{\rm WCT} \nil + a\nil$, because
   $\varepsilon$ is a completed trace of the first term but not of the second.
   For $\precsim_{\rm WF}$ this requirement is in fact \emph{needed},
   in the sense that the version of $\precsim_{\rm WF}$ with this
   requirement can be obtained as the congruence-closure for the
   $+$-operator of the version from \cite{Gla93} without this requirement.
   A similar observation can be made for $\precsim_{\rm WIF}$.
   For $\precsim_{\rm WCT}$ the version with the requirement,
   presented here, appears to be strictly coarser than the
   BCCS-congruence closure of the version from \cite{Gla93} without this requirement.
   As it is out of the scope of this paper to characterize this
   congruence closure, which may be most deserving of the name
   \emph{weak completed trace preorder}, here we simply employ
   $\precsim_{\rm WCT}$ as defined above.

   The requirement that $\mc{WT}(s_1)=\mc{WT}(s_2)$ is needed to make $\precsim_{\rm WIF}$
   a precongruence for BCCS($A$). Without this requirement we would e.g.\ have $\tau a\nil
   \precsim_{\rm WIF} \tau a\nil + b\nil$. In particular, $(\varepsilon,\{b\})$
   is an impossible future not only of the first but also of the second term, because $\tau a\nil + b\nil\mv{\tau}a\nil$.
   However, $c(\tau a\nil) \not\precsim_{\rm WIF} c(\tau a\nil + b\nil)$, because
   $(\varepsilon,\{cb\})$ is an impossible future of the first but not of the second term.
\end{rem}

\subsection{Axiomatization}
An \emph{axiomatization} is a collection of equations $t \approx
u$ or of inequations $t \preccurlyeq u$. The (in)equations in an
axiomatization $E$ are referred to as \emph{axioms}.  If $E$ is an
equational axiomatization, we write $E\vdash t\approx u$ if the
equation $t\approx u$ is derivable from substitution instances of the axioms in $E$
using the rules of equational logic (reflexivity, symmetry, transitivity,
and closure under contexts), i.e.,
\[
\frac{~}{t \approx t}\qquad  \frac{t \approx u}{u \approx t}  \qquad
\frac{t \approx u \hspace*{3mm} u \approx v}{t \approx v} \qquad
\frac{t\approx u}{\alpha t\approx \alpha u} \qquad
\frac{t_1 \approx u_1 \hspace*{3mm} t_2 \approx u_2}{t_1+t_2 \approx u_1+u_2}
\]
For the derivation of an inequation $t\preccurlyeq u$ from an
inequational axiomatization $E$, denoted by $E\vdash t\preccurlyeq
u$, the rule for symmetry is omitted. We will also allow equations
$t\approx u$ in inequational axiomatizations, as an abbreviation
of two separate equations $t\preccurlyeq u$ and $u\preccurlyeq t$.

An axiomatization $E$ is \emph{sound} modulo a preorder
$\sqsubseteq$ (or equivalence $\equiv$) if for any terms $t,u$,
from $E\vdash t\preccurlyeq u$ (or $E\vdash t\approx u$) it
follows that $\rho(t)\sqsubseteq\rho(u)$ (or
$\rho(t)\equiv\rho(u)$) for all closed substitutions $\rho$. $E$
is \emph{ground-complete} for $\sqsubseteq$ (or $\equiv$) if for
any closed terms $p,q$, $p\sqsubseteq q$ (or $p\equiv q$) implies
$E\vdash p\preccurlyeq q$ (or $E\vdash p\approx q$). And $E$ is
\emph{$\omega$-complete} if for any terms $t,u$ with
$E\vdash\rho(t)\preccurlyeq\rho(u)$ (or
$E\vdash\rho(t)\approx\rho(u)$) for all closed substitutions
$\rho$, we have $E\vdash t\preccurlyeq u$ (or $E\vdash t\approx
u$). The equational theory of a process algebra
modulo a preorder $\sqsubseteq$ (or equivalence
$\equiv$) is said to be \emph{finitely based} if there exists a
finite, $\omega$-complete axiomatization that is sound and
ground-complete for the process algebra modulo
$\sqsubseteq$ (or $\equiv$).

The core axioms A1-4 below are sound for BCCS($A$) modulo
every semantics in the spectrum depicted in
Fig.~\ref{fig:spectrum}. We assume that A1-4 are included in every
axiomatization, and write $t=u$ if $\mbox{A1-4}\vdash t\approx u$.
\[
\begin{array}{lrcl}
{\rm A}1~~&x+y &~\approx~& y+x\\
{\rm A}2~~&(x+y)+z &~\approx~& x+(y+z)\\
{\rm A}3~~&x+x &~\approx~& x\\
{\rm A}4~~&x+\nil &~\approx~& x\\
\end{array}
\]

\section{Axiomatizability: From Concrete to Weak Semantics} \label{sec:link}

We present a general method to derive a ground-complete axiomatization for BCCS($A$) modulo a weak semantics
from a ground-complete axiomatization for BCCSP($A$) modulo its concrete counterpart. Moreover, if the original
axiomatization is $\omega$-complete, then so is the resulting axiomatization.

\subsection{Generating an axiomatization for a corresponding weak semantics}

\begin{definition}\label{def:corresponding-weak-equivalence}
Given an equivalence $\equiv_c$ (resp.\ preorder $\sqsubseteq_c$) which is a (pre)congruence for BCCSP($A$).
A \emph{corresponding weak} equivalence $\equiv_w$ (resp.\ preorder $\sqsubseteq_w$) is a (pre)congruence for BCCS($A$)
that coincides with $\equiv_c$ (resp.\ $\sqsubseteq_c$) over closed BCCSP($A$) terms.
\end{definition}

$\precsim_{\rm WT}$, $\precsim_{\rm WCT}$, $\precsim_{\rm WF}$ and $\precsim_{\rm WIF}$ are
corresponding weak preorders of $\precsim_{\rm T}$, $\precsim_{\rm CT}$, $\precsim_{\rm F}$
and $\precsim_{\rm IF}$, respectively. Likewise for the four associated equivalences.

Consider an axiomatization that is ground-complete for BCCSP($A$) modulo a concrete (pre)congruence relation.
We present an algorithm to generate a ground-complete axiomatization for BCCS($A$) modulo a corresponding weak semantics.
The algorithm prescribes the presence of the following two axioms. Actually,
an instance of WIF1 is supposed to be present for each $\alpha\in A_\tau$.
\[
\begin{array}{lrcl}
 {\rm WIF}1~~   &  \alpha(\tau x+\tau y) &\approx& \alpha x+\alpha y      \\
 {\rm WIF}2~~    & \tau x+y &\approx& \tau x+\tau(x+y)
\end{array}
\]
WIF1-2 make it possible to eliminate all non-initial occurrences of $\tau$ within a term (see Prop.~\ref{prop:nf}).
These two axioms---and hence our algorithm---are sound only for semantics at least as coarse as weak impossible futures semantics.
In particular, they are sound for weak failures, completed trace and trace semantics
(cf.\ Fig.~\ref{fig:spectrum}).

In case of a weak corresponding preorder $\sqsubseteq_w$, the
algorithm may moreover prescribe the presence of two axioms concerning
the introduction or elimination of initial occurrences of $\tau$'s, which
are needed to make the weak preorder under consideration a precongruence
(cf.\ Remark~\ref{rem:precongruence}).
\[
\begin{array}{lrcl}
 {\rm W}1~~   &  x &\preccurlyeq& \tau x    \\
 {\rm W}2~~    & \tau x &\preccurlyeq& x
\end{array}
\]
W1 must be present if $p\sqsubseteq_w q$ for some closed BCCS($A$) terms $p$ and $q$ with $p\not\mv{\tau}$ and $q\mv{\tau}$.
Likewise, W2 must be present if $p\sqsubseteq_w q$ for some closed BCCS($A$) terms $p$ and $q$ with $p\mv{\tau}$ and $q\not\mv{\tau}$.
W1 is sound for $\precsim_{\rm WIF}$, so also for $\precsim_{\rm WF}$, $\precsim_{\rm WCT}$ and $\precsim_{\rm WT}$
(cf.\ Fig.~\ref{fig:spectrum}).
And W2 is sound for $\precsim_{\rm WT}$, while for the other three weak preorders $p\sqsubseteq_w q$ and $p\mv{\tau}$ imply $q\mv{\tau}$.

In case of a weak corresponding equivalence $\equiv_w$, the algorithm
prescribes the presence of the axiom
\[
\begin{array}{lrcl}
 {\rm WE}~~   &  x &\approx& \tau x    \\
\end{array}
\]
if $p\equiv_w q$ for some closed BCCS($A$) terms $p$ and $q$ with $p\not\mv{\tau}$ and $q\mv{\tau}$.
WE is sound for $\equiv_{\rm WT}$, while $\equiv_{\rm WIF}$, $\equiv_{\rm WF}$ and $\equiv_{\rm WCT}$ do not require the presence of WE.

Furthermore, the algorithm uses an operator called ${\it init\mbox{-}\tau}$ that maps BCCSP($A$) terms to BCCS($A$) terms
by renaming initial actions into $\tau$. It is defined inductively by:
\[
\begin{array}{lcl}
{\it init\mbox{-}\tau}(\nil) &=& \nil\\
{\it init\mbox{-}\tau}(t+u) &=& {\it init\mbox{-}\tau}(t)+{\it init\mbox{-}\tau}(u)\\
{\it init\mbox{-}\tau}(at) &=& \tau t\\
{\it init\mbox{-}\tau}(x) &=& x
\end{array}
\]
This operator lifts to (in)equations and axiomatizations as expected.

Now we are ready to formulate how an axiomatization $E$ for BCCSP($A$) modulo a concrete semantics
is transformed into an axiomatization $\mc{A}(E)$ for BCCS($A$) modulo a corresponding weak semantics.
First we treat the case of preorders.

\begin{definition}\label{def:algorithm-weak-preorder}
Let $E$ be an axiomatization for BCCSP($A$) modulo a preorder $\sqsubseteq_c$.
The axiomatization $\mc{A}(E)$ for BCCS($A$) modulo a corresponding weak preorder $\sqsubseteq_w$ consists of
the following inequations:

\begin{enumerate}
  \item \label{alg-req1} $E$.

  \item \label{alg-req2} ${\it init\mbox{-}\tau}(E)$.

  \item \label{alg-req3} WIF1-2.

  \item \label{alg-req4} If $p\sqsubseteq_w q$ for some closed BCCS($A$) terms $p$ and $q$ with
  $p\not\mv{\tau}$ and $q\mv{\tau}$, then W1 is included.

  \item \label{alg-req5} If $p\sqsubseteq_w q$ for some closed BCCS($A$) terms $p$ and $q$ with
  $p\mv{\tau}$ and $q\not\mv{\tau}$, then W2 is included.
\end{enumerate}
\end{definition}

\noindent It is essential for the correctness of this approach that axioms do
not mix initial and non-initial occurrences of variables.
An example of such an (unsafe) inequation is $x\preccurlyeq ax$.

\begin{definition}
A term is said to be \emph{safe} if no variable
has both an initial and a non-initial occurrence in it.
An (in)equation $t\preccurlyeq u$ or $t\approx u$ is safe if $t+u$ is safe.
An axiomatization is safe if all its axioms are so.
\end{definition}

\begin{thm} \label{thm:correct}
   Let $\sqsubseteq_c$ be a precongruence for BCCSP($A$) and $\sqsubseteq_w$ a corresponding weak preorder.
   Let $E$ be a ground-complete axiomatization for BCCSP($A$) modulo $\sqsubseteq_c$, which is safe and contains A1-4.
   Then $\mc{A}(E)$ is ground-complete for BCCS($A$) modulo $\sqsubseteq_w$.
   If $E$ is moreover $\omega$-complete for BCCSP($A$), then $\mc{A}(E)$ is $\omega$-complete for BCCS($A$).
\end{thm}

\noindent
Note that this theorem does not address the soundness of $\mc{A}(E)$ for BCCS($A$) modulo $\sqsubseteq_w$.
This is left to the user as a separate proof obligation.

The basic ideas behind the method above are as follows.
With WIF1-2, each BCCS($A$) term can be equated to either a BCCSP($A$) term or
a term $\sum_{i\in I}\tau t_i$ where the $t_i$ are BCCSP($A$) terms (see Prop.~\ref{prop:nf}).
And with the axioms in ${\it init\mbox{-}\tau}(E)$,
a derivation of $\sum_{i\in I}at_i\preccurlyeq\sum_{j\in J}au_j$ from $E$ can be converted
into a derivation of $\sum_{i\in I}\tau t_i\preccurlyeq\sum_{j\in J}\tau u_j$ from $\mc{A}(E)$.
These constitute key steps in the proof of Thm.~\ref{thm:correct}.

The algorithm to generate an axiomatization for BCCS($A$) modulo a weak equivalence from an axiomatization
for BCCSP($A$) modulo the corresponding concrete equivalence can be adapted accordingly from the algorithm
for preorders.

\begin{definition}\label{def:algorithm-weak-equivalence}
Let $E$ be an axiomatization for BCCSP($A$) modulo an equivalence $\equiv_c$.
The axiomatization $\mc{A}(E)$ for BCCS($A$) modulo a corresponding weak equivalence $\equiv_w$
consists of the following equations:
\pagebreak[3]

\begin{enumerate}
  \item $E$.

  \item ${\it init\mbox{-}\tau}(E)$.

  \item WIF1-2.

  \item If $p\equiv_w q$ for some closed BCCS($A$) terms $p$ and $q$ with $p\mv{\tau}$ and $q\not\mv{\tau}$,
  then WE is included.
\end{enumerate}
\end{definition}

\begin{thm} \label{thm:correct2}
   Let $\equiv_c$ be a congruence for BCCSP($A$) and $\equiv_w$ a corresponding weak equivalence.
   Let $E$ be a ground-complete axiomatization for BCCSP($A$) modulo $\equiv_c$, which is safe and contains A1-4.
   Then $\mc{A}(E)$ is ground-complete for BCCS($A$) modulo $\equiv_w$.
   If $E$ is moreover $\omega$-complete for BCCSP($A$), then $\mc{A}(E)$ is $\omega$-complete for BCCS($A$).
\end{thm}

\subsection{Correctness of the transformations}

We establish the correctness of the algorithms for preorders and equivalences.
We will only prove Thm.~\ref{thm:correct} for preorders, as Thm.~\ref{thm:correct2} for equivalences can be proved following the same lines.
For a start, we show that the following equations can be derived from A1-4+WIF1-2:
\[
 \begin{tabular}{l@{\qquad}rcl}
 D1 & $\tau(\tau x+y)$  & $\approx$ & $\tau x+ y$\vspace{1mm}\\
 D2 & $\alpha({\displaystyle\sum_{i\in I}\tau x_i +y)}$ & $\approx$ & ${\displaystyle\sum_{i\in I} \alpha x_i+\alpha(\sum_{i\in I}x_i+y)}$
\end{tabular}
\]

\begin{lemma} \label{lemma5}
  \mbox{\rm D1-2} are derivable from \mbox{\rm A1-4+WIF1-2}.
\end{lemma}
\begin{proof}

For D1,
  \[
   \begin{array}{rllr}
           \tau(\tau x+y)  & \approx &  \tau(\tau x+\tau(x+y)) & ({\rm WIF2})\\
                      & \approx &  \tau x+\tau(x+y) & ({\rm WIF1})\\
                      & \approx & \tau x+y & ({\rm WIF2})\\
  \end{array}
  \]

\medskip

\noindent
For D2, we apply induction on $|I|$. The base case $I\mathbin=\emptyset$ is trivial.
For $|I|\mathbin\geq 1$, pick an $i_0 \mathbin\in I$.\\[2ex]
\mbox{}\hfill
$  \begin{array}[b]{rllr}
    {\displaystyle  \alpha(\sum_{i\in I}\tau x_i +y)} & = & \alpha(\tau x_{i_0} + {\displaystyle \sum_{i\in I\setminus\{i_0\}} }\tau x_i +    y) &\vspace{1mm}\\
                                       & \approx & \alpha (\tau x_{i_0}+\tau (x_{i_0} + {\displaystyle\sum_{i\in I\setminus\{i_0\}}\tau x_i +y) }) & ({\rm WIF2})\vspace{1mm}\\
                                       & \approx & \alpha x_{i_0} + \alpha(x_{i_0} + {\displaystyle\sum_{i\in I\setminus\{i_0\}}\tau x_i +y)} & ({\rm WIF1})\vspace{1mm}\\
                                       & \approx & \alpha x_{i_0} + {\displaystyle \sum_{i\in I\setminus\{i_0\}} \alpha x_i + \alpha(x_{i_0} + \sum_{i\in I\setminus\{i_0\}}  x_i +y)} & ({\rm induction})\vspace{1mm}\\
                                       & = & {\displaystyle \sum_{i\in I} \alpha x_i + \alpha(\sum_{i\in I}x_i+y)}&
  \end{array}$
\end{proof}

\noindent The following proposition on the elimination of $\tau$'s from BCCS($A$) terms will play a key role
in the proof of Thm.~\ref{thm:correct}.

\begin{proposition}\label{prop:nf}
Let $t$ be a BCCS($A$) term.
\begin{enumerate}
  \item If $t\not\mv{\tau}$, then A1-4+WIF1-2 $\vdash t\approx t'$ for some BCCSP($A$) term $t'$.
  \item If $t\mv{\tau}$, then A1-4+WIF1-2 $\vdash t\approx \sum_{i\in I}\tau t_i$
  where $I\neq\emptyset$ and the $t_i$ are BCCSP($A$) terms.
\end{enumerate}
\end{proposition}

\begin{proof}
It is easy to see that A1-4 and WIF1-2 equate only terms of equal weak depth.
For convenience, ``$\mbox{A1-4+WIF1-2} \vdash$'' is omitted here.
\begin{enumerate}
  \item We apply induction on the weak depth of $t$.
  Since $t\not\mv{\tau}$, $t=\sum_{i\in I}a_i t_i + \sum_{j\in J}x_j$. For each $i\in I$,
  \[  t_i \approx\sum_{k\in K_i} \tau t_k' + \sum_{\ell\in L_i} b_\ell t_\ell' + \sum_{m\in M_i}y_m\]
  Moreover, by means of D1 we can guarantee that for each $k\in K_i$, $t_k'\not\mv{\tau}$.
  So, by D2,
  \[ a_it_i \approx \sum_{k\in K_i} a_i t_k' + a_i(\sum_{k\in K_i} t_k'+ \sum_{\ell\in L_i} b_\ell t_\ell' + \sum_{m\in M_i}y_m)\]
  For each $k\in K_i$, by induction, $t_k'\approx t_k''$ where $t_k''$ is a BCCSP($A$) term.
  Likewise, by induction, $\sum_{k\in K_i}  t_k'+ \sum_{\ell\in L_i} b_\ell t_\ell' + \sum_{m\in M_i}y_m \approx t'''_i$ where $t'''_i$ is a BCCSP($A$) term.
  Hence,
  \[  t = \sum_{i\in I} a_it_i + \sum_{j\in J}x_j \approx \sum_{i\in I} (\sum_{k\in K_i} a_i t_k'' + a_it'''_i) + \sum_{j\in J}x_j \]
  And this last term is in BCCSP($A$).

  \medskip

  \item Since $t\mv{\tau}$,
  $t\approx \sum_{i\in I}\tau t_i+\sum_{j\in J} a_j t_j+\sum_{k\in K}x_k$
  where $I\neq\emptyset$. Moreover, by means of D1 we can guarantee that for each $i\in I$, $t_i\not\mv{\tau}$.
  Pick an $i_0\in I$. By WIF2,
    \[   t \approx   \sum_{i\in I}\tau t_i + \tau(t_{i_0} + \sum_{j\in J} a_jt_j + \sum_{k\in K} x_k) \]
  And by (1), the terms $t_i$ and $a_jt_j$ can all be equated to BCCSP($A$) terms.
\qedhere
\end{enumerate}
\end{proof}

\noindent We proceed to prove that a derivation of $t\preccurlyeq u$ from $E$ yields a derivation of
${\it init\mbox{-}\tau}(t\preccurlyeq u)$ from ${\it init\mbox{-}\tau}(E)$. First we establish a lemma
as a stepping stone toward this result.

\begin{lemma}\label{lem:init-tau}
Let $\sigma(t)$ be a BCCSP($A$) term, and let $t$ be safe. Then
\[
{\it init\mbox{-}\tau}(\sigma(t))=\sigma'({\it init\mbox{-}\tau}(t))
\]
where $\sigma'(x)={\it init\mbox{-}\tau}(\sigma(x))$ if $x$ has an initial occurrence in $t$ and $\sigma'(x)=\sigma(x)$ otherwise.
\end{lemma}

\begin{proof}
By induction on the structure of $t$.
\begin{itemize}
\item
$t=\nil$: We have ${\it init\mbox{-}\tau}(\sigma(\nil))=\nil=\sigma'({\it init\mbox{-}\tau}(\nil))$.

\item
$t=x$: We have ${\it init\mbox{-}\tau}(\sigma(x))=\sigma'(x)=\sigma'({\it init\mbox{-}\tau}(x))$.

\item
$t=t_1+t_2$: For $n=1,2$ we define $\sigma_n(x)={\it init\mbox{-}\tau}(\sigma(x))$ if $x$ has an initial occurrence in $t_n$ and $\sigma_n(x)=\sigma(x)$ otherwise.
Since $t_1+t_2$ is safe, $\sigma'$ and $\sigma_n$ coincide over $\var(t_n)$ for both $n=1$ and $n=2$. Hence,
\[
\begin{array}{rcll}
{\it init\mbox{-}\tau}(\sigma(t_1+t_2)) &=& {\it init\mbox{-}\tau}(\sigma(t_1))+{\it init\mbox{-}\tau}(\sigma(t_2))\vspace{2mm}\\
&=& \sigma_1({\it init\mbox{-}\tau}(t_1))+\sigma_2({\it init\mbox{-}\tau}(t_2)) & \mbox{(by induction)}\vspace{2mm}\\
&=& \sigma'({\it init\mbox{-}\tau}(t_1))+\sigma'({\it init\mbox{-}\tau}(t_2))\vspace{2mm}\\
&=& \sigma'({\it init\mbox{-}\tau}(t_1+t_2))
\end{array}
\]

\item
$t=at'$: Since $\sigma$ and $\sigma'$ coincide over $\var(t')$,\\[1.5ex]
\mbox{}\hfill
$
{\it init\mbox{-}\tau}(\sigma(at'))={\it init\mbox{-}\tau}(a\sigma(t'))=\tau\sigma(t')=\tau\sigma'(t')=\sigma'(\tau t')=\sigma'({\it init\mbox{-}\tau}(at'))
$
\qedhere
\end{itemize}
\end{proof}

\begin{proposition}\label{prop:key}
Let $E$ be a safe axiomatization for BCCSP($A$) and suppose that $E\vdash t\preccurlyeq u$. Then
\[
{\it init\mbox{-}\tau}(E)\vdash{\it init\mbox{-}\tau}(t\preccurlyeq u)
\]
\end{proposition}

\begin{proof}
By induction on the derivation of $t\preccurlyeq u$ from $E$. The cases of reflexivity, transitivity and closure under context are straightforward.
We focus on the case of a substitution instance of an axiom, meaning that there are an axiom $v\preccurlyeq w$ in $E$ and a
substitution $\sigma$ such that $\sigma(v)=t$ and $\sigma(w)=u$. Since $v$ and $w$ are safe, Lem.~\ref{lem:init-tau} can be applied to both $v$ and $w$.
We define $\sigma'(x)={\it init\mbox{-}\tau}(\sigma(x))$ if $x$ has an initial occurrence in $v+w$ and $\sigma'(x)=\sigma(x)$ otherwise.
Since $v\preccurlyeq w$ is safe, $\sigma'$ can be used in the application of Lem.~\ref{lem:init-tau} to both $v$ and $w$.
By two applications of Lem.~\ref{lem:init-tau} and one of the axiom ${\it init\mbox{-}\tau}(v\preccurlyeq w)$ in ${\it init\mbox{-}\tau}(E)$
we derive\\[1.5ex]
\mbox{}\hfill
$\displaystyle
{\it init\mbox{-}\tau}(t) = {\it init\mbox{-}\tau}(\sigma(v))
= \sigma'({\it init\mbox{-}\tau}(v))
\preccurlyeq \sigma'({\it init\mbox{-}\tau}(w))
= {\it init\mbox{-}\tau}(\sigma(w))
= {\it init\mbox{-}\tau}(u)
$
\end{proof}

Now we are ready to prove Thm.~\ref{thm:correct} for preorders. As said before, the case of equivalences can be proved following the same lines.

\begin{proof}
Let $E$ be a ground-complete axiomatization for BCCSP($A$) modulo $\sqsubseteq_c$.
Suppose that $t\sqsubseteq_w u$ where either the BCCS($A$) terms $t$ and $u$ are closed or $E$ is $\omega$-complete for BCCSP($A$).
To show that $\mc{A}(E)\vdash t\preccurlyeq u$, we distinguish four cases.
Note, with regard to Prop.~\ref{prop:nf}, that A1-4 and WIF1-2 equate closed terms only with closed terms.

\begin{enumerate}
  \item $t\not\mv{\tau}$ and $u\not\mv{\tau}$.
  By Prop.~\ref{prop:nf}(1), from $\mbox{A1-4+WIF1-2}$ we can derive $t\approx t'$ and $u\approx u'$ where
  $t'$ and $u'$ are BCCSP($A$) terms (and closed if $t$ and $u$ are closed).
  Since $t\sqsubseteq_w u$ and $\mbox{A1-4+WIF1-2}$ are sound for BCCS($A$) modulo $\equiv_w$, we have $t'\sqsubseteq_w u'$.
  Since $t'$ and $u'$ are BCCSP($A$) terms and $\sqsubseteq_w$ coincides with $\sqsubseteq_c$ over closed BCCSP($A$) terms, $t'\sqsubseteq_c u'$.
  Since $E$ is ground-complete for BCCSP($A$) modulo $\sqsubseteq_c$, and $\omega$-complete if $t$ or $u$ is not closed,
  it follows that $E\vdash t'\preccurlyeq u'$. Hence $\mc{A}(E)\vdash t\preccurlyeq u$.

  \item $t\mv{\tau}$ and $u\mv{\tau}$.
  By Prop.~\ref{prop:nf}(2), from $\mbox{A1-4+WIF1-2}$ we can derive $t\approx \sum_{i\in I} \tau t_i$
  and $u\approx \sum_{j\in J} \tau u_j$ where $I,J\neq\emptyset$ and the $t_i$ and $u_j$ are BCCSP($A$) terms (and closed if $t$ and $u$ are closed).
  Since $t\sqsubseteq_w u$ and $\mbox{A1-4+WIF1-2}$ are sound for BCCS($A$) modulo $\equiv_w$,
  we have $\sum_{i\in I} \tau t_i \sqsubseteq_w \sum_{j\in J} \tau u_j$. Pick an $a\in A$.
  We have $\sum_{i\in I} at_i \equiv_{\rm WIF} a(\sum_{i\in I} \tau t_i) \sqsubseteq_w a(\sum_{j\in J} \tau u_j) \equiv_{\rm WIF} \sum_{j\in J} au_j$.
So $\sum_{i\in I} at_i \sqsubseteq_w \sum_{j\in J} au_j$.
Since the $t_i$ and $u_j$ are BCCSP($A$) terms and $\sqsubseteq_w$ coincides with $\sqsubseteq_c$ over closed BCCSP($A$) terms,
$\sum_{i\in I} at_i \sqsubseteq_c \sum_{j\in J} au_j$.
Since $E$ is ground-complete for BCCSP($A$) modulo $\sqsubseteq_c$, and $\omega$-complete if $t$ or $u$ is not closed,
it follows that $E\vdash \sum_{i\in I} at_i\preccurlyeq \sum_{j\in J} au_j$. So, since $E$ is safe, by Prop.~\ref{prop:key},
$\mc{A}(E)\vdash \sum_{i\in I} \tau t_i={\it init\mbox{-}\tau}(\sum_{i\in I} at_i)\preccurlyeq{\it init\mbox{-}\tau}(\sum_{j\in J} au_j)=\sum_{j\in J} \tau u_j$.
Hence $\mc{A}(E)\vdash t\preccurlyeq u$.

  \item $t\not\mv{\tau}$ and $u\mv{\tau}$. By requirement (\ref{alg-req4}) of Def.~\ref{def:algorithm-weak-preorder},
  W1 is included in $\mc{A}(E)$.
  We have $\mc{A}(E)\vdash t\preccurlyeq \tau t \preccurlyeq \tau u\approx u$. The first step follows from W1,
  the second from case (2), and the third from D1 together with $u\mv{\tau}$.

  \item $t\mv{\tau}$ and $u\not\mv{\tau}$. By requirement (\ref{alg-req5}) of Def.~\ref{def:algorithm-weak-preorder},
  W2 is included in $\mc{A}(E)$. We have $\mc{A}(E)\vdash t\approx \tau t \preccurlyeq \tau u\preccurlyeq u$.
  The first step follows from D1 together with $t\mv{\tau}$, the second from case (2), and the third from W2.
\qedhere
\end{enumerate}
\end{proof}

\begin{rem}\label{rem:alternative}
There is an alternative approach for the method introduced in this section that avoids the use of the ${\it init\mbox{-}\tau}$
operator. That is, clause (\ref{alg-req2}) in the construction of $\mc{A}(E)$ in Def.~\ref{def:algorithm-weak-preorder} is omitted.
Moreover, the axiomatization $E$ does not have to be safe. 
And while we chose to ignore preservation of soundness, as the
${\it init\mbox{-}\tau}$ operator would give rise to some technical
complications, in the alternative approach, soundness of $E$ modulo
the strong semantics yields soundness of $\mc{A}(E)$ modulo the
weak semantics in a straightforward fashion, provided we assume the
strong preorder (or equivalence) is included in the corresponding weak one.
However, the price to pay is that this alternative method only works
for ground-completeness (so not for $\omega$-completeness), and
assumes the so-called Fresh Atom Principle (see
e.g.\ \cite{Gla05}). We briefly sketch the idea behind this
alternative approach.

The axiomatization $E$ is required to be sound and ground-complete for BCCSP($A$). The crucial case (2) in the proof of
Thm.~\ref{thm:correct} can now be tackled without ${\it init\mbox{-}\tau}$. As before we arrive at
$\sum_{i\in I} \tau t_i\sqsubseteq_w \sum_{j\in J} \tau u_j$, but as we are dealing with ground-completeness
only, the $t_i$ and $u_j$ are closed terms. Let $a$ be a fresh action which is not in the alphabet $A$.
The last mentioned relation together with the soundness of WIF1 yields $\sum_{i\in I} at_i\sqsubseteq_w \sum_{j\in J} au_j$.
This implies $\sum_{i\in I} at_i\sqsubseteq_c \sum_{j\in J} au_j$.
Since $a$ is fresh, renaming it into $\tau$ yields
$\sum_{i\in I} \tau t_i \sqsubseteq_c \sum_{j\in J} \tau u_j$, where $\tau$ is interpreted as a concrete action.
So by ground-completeness,
$E\vdash \sum_{i\in I} \tau t_i \preccurlyeq \sum_{j\in J} \tau u_j$, which implies
$\mc{A}(E)\vdash t \preccurlyeq u$.

To see why this reasoning does not extend to $\omega$-completeness, assume that the $t_i$ and $u_j$
are open terms. Then the inequation $\sum_{i\in I} \tau t_i\sqsubseteq_w \sum_{j\in J} \tau u_j$
really means that $\sum_{i\in I} \tau\sigma(t_i)\sqsubseteq_w \sum_{j\in J} \tau\sigma(u_j)$ for any closed substitution
$\sigma$ in BCCS($A$). From that, one may not conclude that this equation also holds for any closed substitution
$\sigma$ in BCCS($A \cup\{a\}$), and the latter is needed to infer $\sum_{i\in I} at_i\sqsubseteq_c \sum_{j\in J} au_j$.
\end{rem}

\section{Application to Failures, Completed Traces and Traces}
\label{sec:apps}

In this section, the algorithm from the previous section is applied to produce axiomatizations for BCCS($A$)
modulo the weak failures, completed trace and trace preorders and equivalences from axiomatizations for BCCSP($A$)
modulo their concrete counterparts.


\subsection{Failures semantics}

According to \cite{CFLN08}, A1-4 together with the axiom
  \[{\rm F1}\qquad a(x+y) \preccurlyeq  ax+a(y+z)\]
constitute a sound and ground-complete axiomatization for BCCSP($A$)
modulo $\precsim_{\rm F}$. If $|A|=\infty$ then this axiomatization is $\omega$-complete,
while if $|A|<\infty$ then a finite basis for the inequational theory of BCCSP($A$) modulo $\precsim_{\rm F}$ is obtained
by adding the following axiom:
\[\mathrm{F2}\qquad \sum_{a\in A} ax_a \preccurlyeq \sum_{a\in A}ax_a+y.\]

Our algorithm from the previous section produces a ground-complete axiomatization
for BCCS($A$) modulo $\precsim_{\rm WF}$, which consists of A1-4, WIF1-2 and W1 together with
\[
  \begin{array}{r@{\qquad}rcl}
    {\rm F}1'  & \alpha(x+y) &\preccurlyeq&  \alpha x+\alpha(y+z)      \\
  \end{array}
\]
It is not hard to see that this axiomatization is sound modulo $\precsim_{\rm WF}$.
F1 is extended to F1$'$ (allowing initial $\tau$'s) in light of clause (\ref{alg-req2}),
WIF1-2 are included in light of clause (\ref{alg-req3}),
and W1 in light of clause (\ref{alg-req4}) of Def.~\ref{def:algorithm-weak-preorder}.

The axiomatization can be simplified: the following axioms together with A1-4 and WIF1 suffice.
\[
  \begin{array}{l@{\qquad}rcl}
    {\rm WIF}2'   &   \tau(x+y)  & \preccurlyeq  & \tau x+y     \\
    {\rm W}1'   &    x  &\preccurlyeq & \tau x+y                    \\
  \end{array}
\]
On the one hand, WIF2$'$ and W1$'$ are clearly sound modulo $\precsim_{\rm WF}$.
On the other hand, W1 follows directly from W1$'$ (taking $y=\nil$) and A4.
Furthermore, the two directions of WIF2 can be derived as follows:
by WIF2$'$ and A3, $\tau(x+y)+\tau x\preccurlyeq \tau x+y$;
and by WIF2$'$ (with $y=x$) and W1, $\tau x+y\preccurlyeq\tau x+x+y\preccurlyeq\tau x+\tau(x+y)$.
Finally, F1$'$ can be derived as follows: by W1$'$, $x+y\preccurlyeq \tau x+y+z
\preccurlyeq \tau x+\tau(y+z)$; so by WIF1,
$\alpha(x+y)\preccurlyeq \alpha(\tau x+\tau(y+z))\approx \alpha x+\alpha (y+z)$.

According to Thm.~\ref{thm:correct}, if $|A|=\infty$, then this axiomatization is $\omega$-complete.
And if $1<|A|<\infty$, then F2 and ${\it init\mbox{-}\tau}(\mbox{F2})$ have to be added to make the axiomatization $\omega$-complete.
But the latter inequation, $\sum_{a\in A} \tau x_a \preccurlyeq \sum_{a\in A}\tau x_a+y$, can be derived using W1$'$ and WIF1.

\begin{thm}\label{thm:WF}
  \emph{A1-4+WIF1+WIF2$'$+W1$'$} is sound and ground-complete for BCCS($A$) modulo $\precsim_{\rm WF}$. If $|A|=\infty$,
  then it is also $\omega$-complete.
  If $|A|<\infty$, then the axiomatization becomes $\omega$-complete by adding the (sound) axiom \emph{F2}.
\end{thm}

According to \cite{Gla01}, A1-4 together with the axioms
 \[
    \begin{array}{l@{\qquad}r@{~\approx~}l}
    {\rm FE}1~   &    ax + a(y+z) & ax + a(x+y) + a(y+z)    \\
    {\rm FE}2^*~  &    a(bx+u)+a(by+v) & a(bx+by+u)+a(by+v)     \\
    \end{array}
\]
constitute a sound and ground-complete axiomatization for BCCSP($A$)
modulo  $\simeq_{\rm F}$. As remarked in \cite{FN05}, in the presence
of FE1, axiom FE2$^*$ can be simplified to
 \[
    \begin{array}{l@{\qquad}r@{~\approx~}l}
    {\rm FE}2~  &    a(x+by)+a(x+by+bz) & a(x+bx+bz).    \\
    \end{array}
\]
Moreover, in \cite{FN05} it was proved that if $|A|=\infty$ then this axiomatization is $\omega$-complete,
while if $|A|<\infty$ then a finite basis for the inequational theory of BCCSP($A$) modulo $\precsim_{\rm F}$ is obtained
by adding the following axiom:
\[\mathrm{FE3}\qquad a(x+\sum_{b\in A} bz_b) + a(x+y+\sum_{b\in A} bz_b)
                    \approx a(x+y+\sum_{b\in A} bz_b).\]

Our algorithm from the previous section produces a ground-complete axiomatization
for BCCS($A$) modulo $\simeq_{\rm WF}$, which consists of A1-4 and WIF1-2 together with
\[
  \begin{array}{r@{\qquad}r@{~\approx~}l}
    {\rm FE}1'~   &    \alpha x + \alpha(y+z) & \alpha x + \alpha(x+y) + \alpha(y+z)    \\
    {\rm FE}2'~  &    \alpha(x+by)+\alpha(x+by+bz) & \alpha(x+bx+bz).    \\
  \end{array}
\]
It is not hard to see that this axiomatization is sound modulo $\simeq_{\rm WF}$.
FE1-2 are extended (allowing initial $\tau$'s) in light of clause
(2) of Def.~\ref{def:algorithm-weak-equivalence}
and WIF1-2 are included in light of clause (3).

Also this axiomatization can be simplified: the following axiom together with A1-4 and WIF1-2 suffices.
\[
  \begin{array}{r@{\qquad}r@{~\approx~}l}
    {\rm WFE}~   &    ax+ \tau(ay+z) & \tau(ax+ay+z)  \\
  \end{array}
\]
On the one hand, WFE is a direct consequence of WIF2 and FE2$'$:
$ax+ \tau(ay+z) = \tau(ay+z) + \tau(ax+ay+z) = \tau(ax+ay+z)$.
On the other hand, the instances of FE1$'$ and FE2$'$ with
$\alpha=\tau$ can be derived with two applications of WIF2 and with
WIF2 and WFE, respectively:

$\tau x + \tau(x+y) + \tau(y+z) \approx \tau x + y + \tau(y+z)
\approx \tau x + \tau(y+z) + \tau(y+y+z) \approx \tau x + \tau(y+z)$;

$\tau(x+by)+\tau(x+by+bz) \approx \tau(x+by)+bz \approx \tau(x+bx+bz)$.

\noindent
The general instances of FE1$'$ and FE2$'$ now follow with WIF1.

In \cite{Gla97} it was already remarked that {A1-4+WIF1-2+WFE} is
a sound and ground-complete axiomatization for BCCS($A$) modulo $\simeq_{\rm WF}$. 

According to Thm.~\ref{thm:correct2}, if $|A|=\infty$, then this axiomatization is $\omega$-complete.
And if $1<|A|<\infty$, then FE3 and ${\it init\mbox{-}\tau}(\mbox{FE3})$ have to be added to make the axiomatization $\omega$-complete.
But from the latter inequation,
\[\mathrm{FE3'}\qquad \tau(x+\sum_{b\in A} bz_b) + \tau(x+y+\sum_{b\in A} bz_b)
                    \approx \tau(x+y+\sum_{b\in A} bz_b)\]
FE3 can be derived using WIF1.

\begin{thm}\label{thm:WFE}
  \emph{A1-4+WIF1-2+WFE} is sound and ground-complete for BCCS($A$) modulo $\simeq_{\rm WF}$. If $|A|=\infty$,
  then it is also $\omega$-complete.
  If $|A|<\infty$, then the axiomatization becomes $\omega$-complete by adding the (sound) axiom \emph{FE3$'$}.
\end{thm}

In \cite{AFI07,FG07} an algorithm is presented which takes as
input a sound and ground-complete inequational axiomatization
for BCCSP modulo a preorder $\sqsubseteq$ which \emph{includes the
ready simulation preorder} and is \emph{initials
preserving},\footnote{Initials preserving means that $p\sqsubseteq q$ implies
$\mc{I}(p)\subseteq \mc{I}(q)$.} and generates as output an equational
axiomatization which is sound and ground-complete for
BCCSP modulo the corresponding equivalence---its kernel: $\sqsubseteq
\cap \sqsubseteq^{-1}$. Moreover, if the
original axiomatization is $\omega$-complete, then so is the
resulting axiomatization. 
Using this algorithm, the above-mentioned axiomatization of
$\simeq_{\rm F}$ could have been obtained from the one of $\precsim_{\rm F}$.

In \cite{CFG08} we lifted this result to weak semantics, which makes
the aforementioned algorithm applicable to all 87 preorders surveyed
in \cite{Gla93} that are at least as coarse as the ready simulation preorder.
In \cite{CFG09} we obtain an alternative proof of Thm.~\ref{thm:WFE}
by applying the algorithm of \cite{CFG08} to the axiomatizations of Thm.~\ref{thm:WF}.

\subsection{Completed trace semantics}
\label{sec:completed-traces}

According to \cite{Gla01}, A1-4 together with the axiom
\[
    \begin{array}{l@{\qquad}rcl}
    {\rm CTE}~   &   a(bw + cx + y + z) & \approx  &   a(bw + y) + a(cx + z)\\
    \end{array}
\]
constitute a sound and ground-complete axiomatization for BCCSP($A$)
modulo $\simeq_{\rm CT}$. After adding the axiom FE1, the
axiomatization becomes $\omega$-complete \cite{Gro90}.

By applying our algorithm, we obtain a ground-complete axiomatization for BCCS($A$) modulo
$\simeq_{\rm WCT}$, which consists of A1-4, WIF1-2 and
\[
    \begin{array}{l@{\qquad}rcl}
    {\rm CTE}'~   &   \alpha(bw + cx + y + z) & \approx  &   \alpha(bw + y) + \alpha(cx + z)\\
    \end{array}
\]
It is not hard to see that this axiomatization is sound modulo $\simeq_{\rm WCT}$.
CTE is extended (allowing initial $\tau$'s) in light of clause (\ref{alg-req2})
of Def.~\ref{def:algorithm-weak-preorder}.
This axiomatization is also $\omega$-complete, since FE1$'$ can be derived, as shown above.
Moreover, CTE$'$ follows from CTE and WIF1.

\begin{thm}
  \emph{A1-4+WIF1-2+CTE} is sound, ground-complete as well as $\omega$-complete for BCCS($A$) modulo $\simeq_{\rm WCT}$.
\end{thm}

The axiomatization A1-4 together with
\[
    \begin{array}{l@{\qquad}rcl}
    {\rm CT}1~   &    ax & \preccurlyeq  &  ax+y     \\
    {\rm CT}2~   &   a(bw + cx + y + z) & \preccurlyeq  &   a(bw + y) + a(cx + z)\\
    \end{array}
\]
from \cite{Gla01} is sound and ground-complete for BCCSP($A$) modulo $\precsim_{\rm CT}$.
After adding the axiom F1, the axiomatization becomes $\omega$-complete.
This $\omega$-completeness result follows from the
$\omega$-completeness of the above axiomatization for BCCSP($A$) modulo $\simeq_{\rm CT}$.
Namely, suppose that $t\precsim_{\rm CT}u$. If $t$ does not contain a summand of the form $at'$,
then clearly $t$ and $u$ must consist of exactly the same variable summands, so that $t=u$.
And if $t$ contains a summand $at'$, then by CT1, $t\preccurlyeq t+u$. Since $t+u\simeq_{\rm CT} u$,
derivability of $t+u\approx u$ from A1-4, CT1-2, F1 follows from the aforementioned
$\omega$-completeness result for BCCSP($A$) modulo $\simeq_{\rm CT}$,
using that CTE follows from CT1-2, and FE1 from CT1 and F1.

By applying the algorithm, we obtain a ground-complete axiomatization for BCCS($A$) modulo
$\precsim_{\rm WCT}$, which consists of A1-4, WIF1-2, W1 and
\[
    \begin{array}{l@{\qquad}rcl}
    {\rm CT}1'   &    \alpha x & \preccurlyeq  &  \alpha x+y     \\
    {\rm CT}2'   &   \alpha (bw + cx + y + z) & \preccurlyeq  &   \alpha (bw + y) + \alpha (cx + z)\\
    \end{array}
\]
It is not hard to see that this axiomatization is sound modulo $\precsim_{\rm WCT}$.
CT1-2 are extended (allowing initial $\tau$'s) in light of clause (\ref{alg-req2})
of Def.~\ref{def:algorithm-weak-preorder}.

WIF1+WIF2$'$+W1$'$+CT1$'$ together with A1-4 suffice, because
CT2$'$ can be derived using W1 and WIF1:
$\alpha(bw + cx + y + z) \preccurlyeq
\alpha(\tau(bw +y) +\tau( cx  + z)) \approx
\alpha(bw + y) + \alpha(cx + z)$.

We conclude that A1-4+WIF1+ WIF2$'$+W1$'$+ CT1$'$ is
ground-complete for BCCS($A$) modulo $\precsim_{\rm WCT}$.
It is also $\omega$-complete, since F1$'$ can be derived, as shown before.

\begin{thm}
  \emph{A1-4+WIF1+WIF2$'$+W1$'$+CT1$'\!$} is sound, ground-complete and $\omega$-complete for BCCS($A$) modulo $\precsim_{\rm WCT}$.
\end{thm}

\subsection{Trace semantics}

According to \cite{Gla01,Gro90}, A1-4 together with the axiom
\[
    \begin{array}{l@{\qquad}rcl}
    {\rm TE}~   &   ax+ay & \approx  &   a(x + y)\\
    \end{array}
\]
constitute a sound and ground-complete axiomatization for BCCSP($A$)
modulo $\simeq_{\rm T}$. If $|A|>1$ it is even $\omega$-complete \cite{Gro90}.

By applying our algorithm, we obtain a ground-complete axiomatization for
BCCS($A$) modulo $\simeq_{\rm CT}$, consisting of A1-4, WIF1-2, WE and 
\[
    \begin{array}{l@{\qquad}rcl}
    {\rm TE}'~   &   \alpha x + \alpha y & \approx  &   \alpha(x + y)\\
    \end{array}
\]
If $|A|>1$ it is even $\omega$-complete.
It is easy to see that it is sound modulo $\simeq_{\rm CT}$.
Clearly, owing to WE, WIF1-2 are redundant.

\begin{thm}
  \emph{A1-4+WE+TE} is sound and ground-complete for BCCS($A$) modulo $\simeq_{\rm CT}$.
  If $|A|>1$, then this axiomatization is also $\omega$-complete.
\end{thm}

The axiomatization A1-4 together with
\[
    \begin{array}{l@{\qquad}rcl}
    {\rm T}1   &    a(x+  y)  & \preccurlyeq  & ax+ay         \\
    {\rm T}2   &    x & \preccurlyeq  &  x+y     \\
    \end{array}
\]
from \cite{Gla01} is sound and ground-complete for BCCSP($A$) modulo $\precsim_{\rm T}$.
If $|A|>1$, then it is also $\omega$-complete.
This $\omega$-completeness result follows from the $\omega$-completeness of the above axiomatization
for BCCSP($A$) modulo $\simeq_{\rm T}$. Namely, suppose that $t\precsim_{\rm T}u$.
Then by T2, $t\preccurlyeq t+u$. Since $t+u\simeq_{\rm T} u$,
derivability of $t+u\approx u$ from A1-4, T1-2 follows from the aforementioned
$\omega$-completeness result for BCCSP($A$) modulo $\simeq_{\rm T}$,
using that TE follows from T1-2.

By applying the algorithm, we obtain a ground-complete
axiomatization for BCCS($A$) modulo $\precsim_{\rm WT}$, which consists of A1-4, WIF1-2, WE and T2 together with
\[
    \begin{array}{r@{\qquad}rcl}
    {\rm T}1'   &    \alpha(x+y)  & \preccurlyeq  & \alpha x+\alpha y
    \end{array}
\]
It is not hard to see that this axiomatization is sound modulo $\precsim_{\rm WT}$.
T1 is extended in light of clause (\ref{alg-req2}) and WE is introduced in light of
clauses (\ref{alg-req4}) and (\ref{alg-req5}) of Def.~\ref{def:algorithm-weak-preorder}.
Again, owing to WE, WIF1-2 are redundant.

\begin{thm}
  \emph{A1-4+T1$'$+T2+WE} is sound and ground-complete for BCCS($A$) modulo $\precsim_{\rm WT}$.
  If $|A|>1$, then this axiomatization is also $\omega$-complete.
\end{thm}

If $|A|=1$, then $x\preccurlyeq ax$ needs to be added to make the axiomatization $\omega$-complete
(cf.\ \cite{Gro90}). This is the only corner case where our approach breaks down, due to
the fact that this axiom is not safe: the variable $x$ has both an initial and a non-initial occurrence.

Admittedly, the application of Defs.~\ref{def:algorithm-weak-equivalence} and~\ref{def:algorithm-weak-preorder}
to trace semantics is not so interesting, because the axiom WE allows to eliminate all $\tau$'s from
terms, so that the weak setting trivially reduces to the concrete setting.

\section{Ground-Completeness for Impossible Futures} \label{sec:ground}

This section concerns the axiomatizability of impossible futures semantics.

\subsection{Concrete impossible futures preorder}
We prove that A1-4 together with the axioms
\[
\begin{array}{lrcl}
{\rm IF1}~& a(x+y)  & \preccurlyeq & ax+ay\\
{\rm IF2}~ & ax+a(y+z) & \approx & a(x+y)+ax+a(y+z)\\
\end{array}
\]
constitute a sound and ground-complete axiomatization for BCCSP($A$) modulo $\precsim_{\rm IF}$.

\begin{thm} \label{prop16}
\emph{A1-4+IF1-2} is sound and ground-complete for BCCSP($A$) modulo $\precsim_{\rm IF}$.
\end{thm}

It is not hard to see that IF1-2 are sound modulo $\precsim_{\rm IF}$.
For both axioms, the crucial observation for their soundness is that the impossible futures
$(a,B)$ induced by $a(x+y)$ are included in those of $ax$.
To give some intuition for the ground-completeness proof, we first present an example.

\begin{example}\label{ex1}
Let $p=a(a\nil+a^2\nil)+a^4\nil$ and $q=a(a\nil+a^3\nil)+a^3\nil$.
It is not hard to see that $p \precsim_{\rm IF} q$. However,
neither $a(a\nil+a^2\nil) \precsim_{\rm IF} a(a\nil+a^3\nil)$ nor
$a(a\nil+a^2\nil) \precsim_{\rm IF} a^3\nil$ holds. In order to
derive $p \preccurlyeq q$, we therefore first derive $q ~\approx~ p+q$,
and next $p \preccurlyeq p+q$.
\end{example}

In general, to derive a sound closed inequation $p\preccurlyeq q$,
first we derive $q\approx \mb{S}(q)$ (see Lem.~\ref{lem:hatq}),
where $\myoverline{q}$ contains for every $a\in\mc{I}(q)$ a
``saturated'' $a$-summand (see Def.~\ref{def:hatq}). (In
Ex.~\ref{ex1}, this saturated summand would have the form
$a(a\nil+a^2\nil+a^3\nil+a(a\nil+a^2\nil))$.) Then, in the proof
of Thm.~\ref{prop16}, we derive
$\Psi+\myoverline{q}\approx\myoverline{q}$ (equation
(\ref{eqn1})), $p\preccurlyeq \Psi$ (equation (\ref{eqn2})) and
$p\preccurlyeq p+q$ (equation (\ref{eqn3})), where the closed term
$\Psi$ is built from many ``semi-saturated'' summands (like, in
Ex.~\ref{ex1}, $p$). These results together provide the desired
proof (see the last line of the proof of Thm.~\ref{prop16}).

In the remainder of this section, $ap'\Subset p$ denotes that $ap'$ is a summand of $p$.

\begin{definition} \label{def:hatq}
For each closed term $q$, the closed term $\myoverline{q}$ is defined recursively on the depth of $q$ as follows:
\[\myoverline{q} ~=~ q~+~\sum_{a\in\mc{I}(q)} a(\myoverline{\sum_{aq'\Subset q} q'})\]
\end{definition}

\begin{example}
If $q=a(b(c\nil+d\nil)+be\nil)+af\nil$, then $\myoverline{q}=
a(b(c\nil+d\nil)+be\nil)+af\nil +a (b(c\nil+d\nil)+be\nil+f\nil +
b(c\nil+d\nil+e\nil))$.
\end{example}

\begin{lemma} \label{lem:hatq}
  For each closed term $q$, A1-4+IF1-2 $\vdash q \approx  \myoverline{q} $.
\end{lemma}
\begin{proof}
  By induction on the depth of $q$. For any $a\in\mc{I}(q)$,
  \[ q~\approx~ q+a(\sum_{aq'\Subset q} q') ~\approx~ q+a(\myoverline{\sum_{aq'\Subset q} q'})\]
  The first derivation step uses IF2, and the second induction. Hence, summing up over all $a\in\mc{I}(q)$,
  \\\mbox{}\hfill
  $\displaystyle
  q~\approx~ q+\sum_{a\in\mc{I}(q)} a(\myoverline{\sum_{aq'\Subset q} q'}) ~=~ \myoverline{q} ~~~~~~ $
\end{proof}

\noindent For any closed term $q$ and $a\in\mc{I}(q)$, the closed term
$q_a$ is obtained by summing over all closed terms $q'$ such
that $q\mv{a}q'$, and then applying the saturation from
Def.~\ref{def:hatq}. The definition and lemma below generalize this idea to
terms $q_{a_1\cdots a_\ell}$ with $a_1\cdots a_\ell\in\mc{T}(q)$.
The auxiliary terms $q_{a_1\cdots a_\ell}$ will only be
used in the derivation of equation (\ref{eqn1}) within the proof
of Thm.~\ref{prop16}.

\begin{definition} \label{def:Q}
Given a closed term $q$ and a trace $a_1\cdots a_\ell$ of $q$.
\[Q_{a_1\cdots a_\ell}~=~\{q_\ell\mid \mbox{there exists a sequence of transitions }q\mv{a_1}q_{1} \cdots \mv{a_\ell} q_\ell\}\]
and
\[q_{a_1\cdots a_\ell}~=~\myoverline{\sum_{q_\ell\in Q_{a_1\cdots a_\ell}}q_\ell}\]
\end{definition}
Note that $q_\varepsilon=\myoverline{q}$. We prove some basic properties for the terms $q_{a_1\cdots a_\ell}$.

\begin{lemma} \label{lem:Q}
Given a closed term $q$, and a completed trace $a_1\cdots a_d$ of
$q$. Then, for each $0\leq \ell<d$,
\begin{itemize}
\item $q_{a_1\cdots a_\ell} \mv{a_{\ell+1}} q_{a_1\cdots
a_{\ell+1}}$; and \vspace{2mm} \item $q_{a_1\cdots a_\ell}
\mv{a_{\ell+1}} q_{\ell+1}$ ~~~ for all $q_{\ell+1}\in
Q_{a_1\cdots a_{\ell+1}}$.
\end{itemize}
\end{lemma}

\begin{proof}
Clearly, $q_{\ell+1}\in Q_{a_1\cdots a_{\ell+1}}$ iff there exists
some $q_\ell \in Q_{a_1\cdots a_\ell}$ such that
$q_\ell\mv{a_{\ell+1}} q_{\ell+1}$. And since $a_1\cdots
a_{\ell+1}$ is a trace of $q$, $a_{\ell+1}\in\mc{I}(q_\ell)$ for
some $q_\ell \in Q_{a_1\cdots a_\ell}$. So by Def.~\ref{def:hatq},
\[
  q_{a_1\cdots a_\ell}~=~\myoverline{\sum_{q_\ell\in Q_{a_1\cdots a_\ell}} q_\ell} ~\mv{a_{\ell+1}}~
  \myoverline{\sum_{q_{\ell+1} \in Q_{a_1\cdots a_{\ell+1}}} q_{\ell+1}}~=~q_{a_1\cdots a_{\ell+1}}
\]
Moreover, for all $q_{\ell+1}\in Q_{a_1 \cdots a_{\ell+1}}$  we
have $\sum_{q_\ell\in Q_{a_1\cdots a_\ell}} q_\ell \mv{a_{\ell+1}}
q_{\ell+1}$. Hence, by Def.~\ref{def:hatq},
\\[1ex]\mbox{}\hfill
$\displaystyle
[q_{a_1\cdots a_\ell}~=~\myoverline{\sum_{q_\ell\in Q_{a_1\cdots a_\ell}}q_\ell} ~\mv{a_{\ell+1}}~ q_{\ell+1} ~~~~~~ $
\end{proof}

\noindent We now embark on proving the promised ground-completeness result in Thm.~\ref{prop16}.

\begin{proof}
Suppose $p\precsim_{\rm IF} q$. We derive $p\preccurlyeq q$
using induction on the depth of $p$. If $p=\nil$, then clearly
$q=\nil$, and we are done. So assume $p\not=\nil$.

We call a sequence $a_1p_1 \cdots a_k p_k$ a \emph{completed path} of a closed term $p_0$
if $p_0\mv{a_1}p_1 \cdots \mv{a_k}p_k$ with $\mc{I}(p_k)=\emptyset$.
Let $\mc{CP}(p)$ denote the set of completed paths of $p$, ranged over by $\pi$.
Consider any $\pi= a_1p_1\cdots a_d p_d$ in $\mc{CP}(p)$.
Since $p\not=\nil$, we have $d\geq 1$.
We recursively construct closed terms $\psi_\ell^{\pi}$, where $\ell$ counts down
from $d$ to 1. For the base case we define $\psi^{\pi}_d=\nil$. Now
let $1\leq \ell<d$. Since $p \mv{a_1\cdots a_\ell} p_\ell$ and
$p\precsim_{\rm IF} q$, there exists a sequence of transitions $q
\mv{a_1\cdots a_\ell}q_\ell$ such that $\mc{T}(q_\ell)\subseteq
\mc{T}(p_\ell)$. Given the choice of $\pi$ and $\ell$, we pick such a $q_\ell$ and call it $q^\pi_\ell$.
Now define
\[
\psi^{\pi}_\ell  ~=~  q^\pi_\ell + a_{\ell+1} \psi^{\pi}_{\ell+1}
\]
We prove, by induction on $d-\ell$, that for all $1\leq\ell\leq d$,
\[\mc{T}(\psi^{\pi}_\ell)~\subseteq~ \mc{T}(p_\ell)\]
The base case is trivial, since $\mc{T}(\psi^{\pi}_d)=\emptyset$.
Now let $1\leq \ell<d$. By induction, $\mc{T}(
\psi^{\pi}_{\ell+1}) \subseteq \mc{T}(p_{\ell+1})$. Moreover,
$p_\ell\mv{a_{\ell+1}}p_{\ell+1}$, so $\mc{T}(a_{\ell+1}
\psi^{\pi}_{\ell+1}) \subseteq \mc{T}(p_\ell)$. Hence,
$\mc{T}(\psi^{\pi}_\ell) = \mc{T}(q^\pi_\ell) \cup \mc{T}(a_{\ell+1}
\psi^{\pi}_{\ell+1}) \subseteq \mc{T}(p_\ell)$.

\medskip

We now derive three (in)equations that together yield the desired result.
First, we derive by induction on $d-\ell$, for all $1\leq\ell\leq d$,
\[ a_\ell\psi^{\pi}_\ell + q_{a_1\cdots a_{\ell-1}} ~\approx~ q_{a_1\cdots a_{\ell-1}}\]
In the base case, since $\psi^{\pi}_d=\nil\in Q_{a_1\cdots a_d}$
(see Def.~\ref{def:Q}), this is a direct consequence of the second
item in Lem.~\ref{lem:Q}. Now let $1\leq\ell<d$.
\[
\begin{array}{rlr}
 & a_\ell\psi^{\pi}_{\ell} + q_{a_1\cdots a_{\ell-1}} \vspace{2mm} \\
=~& a_\ell(q^\pi_\ell + a_{\ell+1} \psi^{\pi}_{\ell+1}) + q_{a_1\cdots a_{\ell-1}} +~a_\ell q^\pi_\ell + a_\ell q_{a_1\cdots a_\ell} & \mbox{(Lem.~\ref{lem:Q})}\vspace{2mm} \\
\approx~& a_\ell(q^\pi_\ell + a_{\ell+1} \psi^{\pi}_{\ell+1}) +
q_{a_1\cdots a_{\ell-1}}  +~a_\ell q^\pi_\ell + a_\ell(a_{\ell+1} \psi^{\pi}_{\ell+1} + q_{a_1\cdots a_\ell}) & \mbox{(induction)}\vspace{2mm} \\
\approx~& q_{a_1\cdots a_{\ell-1}} + a_\ell q^\pi_\ell  +~a_\ell(a_{\ell+1} \psi^{\pi}_{\ell+1} + q_{a_1\cdots a_\ell}) & \mbox{(IF2)}\vspace{2mm} \\
\approx~& q_{a_1\cdots a_{\ell-1}} + a_\ell q^\pi_\ell   +~a_\ell q_{a_1\cdots a_\ell} & \mbox{(induction)}\vspace{2mm} \\
=~& q_{a_1\cdots a_{\ell-1}} & \mbox{(Lem.~\ref{lem:Q})}
\end{array}
\]
In the end, for $\ell=1$, we have derived $a_1\psi^{\pi}_1 +
q_{\varepsilon} \approx q_{\varepsilon}$. In other words,
\begin{eqnarray}
\label{eqn1}
a_1\psi^{\pi}_1 + \myoverline{q} ~\approx~ \myoverline{q}
\end{eqnarray}

Second, for every $ap'\Subset p$,
\[   p' ~\precsim_{\rm IF}~ \sum_{\pi\in \mc{CP}(a p')}  \psi^{\pi}_1 \]
Namely, consider any (possibly incomplete) path $\pi_0=a_1p_1 \cdots a_hp_h$ of $a p'$.
Extend $\pi_0$ to some $\pi\in\mc{CP}(ap')$.
Clearly, $\psi^{\pi}_{\ell}\mv{a_{\ell+1}}\psi^{\pi}_{\ell+1}$ for all $1\leq
\ell <h$, so $\psi^{\pi}_1 \mv{a_2\cdots a_h} \psi^{\pi}_h$.
Moreover, we proved that $\mc{T}(\psi^{\pi}_h)\subseteq \mc{T}(p_h)$.

So by induction on depth, for every $ap'\Subset p$ we can derive
  \[   p' ~\preccurlyeq~ \sum_{\pi\in \mc{CP}(a p')} \psi^{\pi}_1 \]
And thus, by IF1, we derive
  \[ a p' ~\preccurlyeq~ \sum_{\pi\in \mc{CP}(a p')} a\psi^{\pi}_1 \]
Hence, summing over all summands $ap'$ of $p$, we derive
\begin{eqnarray}
\label{eqn2} p &~\preccurlyeq~& \sum_{a\in\mc{I}(p)}\sum_{ap'\Subset p}\sum_{\pi\in\mc{CP}(ap')} a\psi^{\pi}_1
\end{eqnarray}

Third, since $p\precsim_{\rm IF} q$, clearly, for each $a\in\mc{I}(p)$,
\[
\sum_{ap'\Subset p} p' ~\precsim_{\rm IF}~ \sum_{aq'\Subset q} q'
\]
So by induction on depth, for each $a\in\mc{I}(p)$ we can derive
\[
 \sum_{ap'\Subset p} p' ~\preccurlyeq~ \sum_{aq'\Subset q} q'
\]
So by IF2 and IF1, and since $\mc{I}(p)=\mc{I}(q)$, we derive
\[
   p ~~\approx~~ p + {\displaystyle\sum_{a\in\mc{I}(p)}}
a({\displaystyle\sum_{a p'\Subset p}} p')  \preccurlyeq  p +
{\displaystyle\sum_{a\in\mc{I}(q)}}
a({\displaystyle\sum_{aq'\Subset q}} q') ~~\preccurlyeq~~ p +
{\displaystyle\sum_{a\in \mc{I}(q)}}~
{\displaystyle\sum_{aq'\Subset q}} aq'
\]
That is, we have derived
\begin{eqnarray}
\label{eqn3} p ~\preccurlyeq~ p+q
\end{eqnarray}
Finally, (\ref{eqn3}), Lem.~\ref{lem:hatq}, (\ref{eqn2}) and (\ref{eqn1}) yield the derivation
\\[2ex]\mbox{}\hfill
$\displaystyle
  p ~\preccurlyeq~ p+q ~\approx~ p+\myoverline{q}
~\preccurlyeq~
 \sum_{a\in\mc{I}(p)}~\sum_{ap'\Subset p}~\sum_{\pi\in \mc{CP}(ap')} a\psi^{\pi}_1 + \myoverline{q}
~\approx~ \myoverline{q} ~\approx~ q ~~  $
\end{proof}

\subsection{Weak impossible futures preorder} \label{sec:wifp}

We now apply the link established in Sect.~\ref{sec:link} to
obtain a ground-complete axiomatization for BCCS($A$) modulo $\precsim_{\rm WIF}$.
It consists of A1-4, WIF1-2 and W1 together with
\[
\begin{array}{lrcl}
{\rm IF1}'~& \alpha(x+y)  & \preccurlyeq & \alpha x+\alpha y\\
{\rm IF2}'~ & \alpha x+\alpha(y+z) & \approx & \alpha(x+y)+\alpha x+\alpha(y+z)\\
\end{array}
\]
It is not hard to see that this axiomatization is sound modulo $\precsim_{\rm WIF}$.

Again, this axiomatization can be simplified.
It turns out that IF1$'$ and IF2$'$ are redundant. Namely, by W1 and WIF1,
$\alpha(x+y)\preccurlyeq \alpha(\tau x+\tau y) \approx \alpha x+\alpha y$.
And by WIF1 and WIF2, $\alpha x+\alpha(y+z)\approx\alpha(\tau x+\tau(y+z))
=\alpha(\tau x+\tau(y+z+y)+\tau(y+z))\approx\alpha(\tau x+y+\tau(y+z))
\approx\alpha(\tau(x+y)+\tau x+\tau(y+z))\approx\alpha(\tau(x+y)+\tau(\tau x+\tau(y+z)))
\approx\alpha(x+y)+\alpha(\tau x+\tau(y+z))\approx\alpha(x+y)+\alpha x +\alpha(y+z)$.
Furthermore, WIF2 can be replaced by WIF2$'$.
Namely, by WIF2$'$, $\tau(x+y)+\tau x\preccurlyeq \tau x+y$;
and by WIF2$'$ and W1, $\tau x+y=\tau x+\tau(x+x)+y\preccurlyeq\tau x+x+y \preccurlyeq \tau x+\tau(x+y)$.

The ground-completeness claim in the following corollary is an immediate consequence of
Thm.~\ref{prop16} together with Thm.~\ref{thm:correct}.

\begin{cor}
  \emph{A1-4+WIF1+WIF2$'$+W1} is sound and ground-complete for BCCS($A$) modulo $\precsim_{\rm WIF}$.
\end{cor}

This result was already proved in \cite{VM01} (for a slightly more complicated axiomatization).
There an intricate ground-completeness proof was given, which relies heavily on the presence of $\tau$.

\subsection{Weak impossible futures equivalence}\label{nonax}

We now prove that there does not exist a finite, sound, ground-complete axiomatization for
BCCS($A$) modulo $\simeq_{\rm WIF}$. The cornerstone for
this negative result is the following infinite family of closed equations. Pick an $a\in A$.
For each $m\geq 0$,
\[ \tau a^{2m}\nil+\tau(a^m\nil+a^{2m}\nil) ~\approx~ \tau(a^m\nil+a^{2m}\nil) \]
is sound modulo $\simeq_{\rm WIF}$. We start with a few lemmas.

\begin{lemma} \label{newvariable}
If $t\precsim_{\rm WIF} u$ and $t\Rightarrow\mv{\tau} t'$,
then there is a term $u'$ with $u\Rightarrow\mv{\tau} u'$ and $\var(u')\subseteq \var(t')$.
\end{lemma}

\begin{proof}
Let $t\Rightarrow\mv{\tau} t'$. Pick an $a\in A$ and $m>\depth_w(t)$,
and consider the closed substitution $\rho$ defined by $\rho(x)=\nil$ if
$x\in\var(t')$ and $\rho(x)=a^m\nil$ if $x\not\in\var(t')$. Since
$\rho(t)\Rightarrow\rho(t')$ with $\depth_w(\rho(t'))=\depth_w(t')<m$, and
$\rho(t)\precsim_{\rm WIF}\rho(u)$, clearly $\rho(u)\Rightarrow q$
for some $q$ with $\depth_w(q)<m$. From the definition of $\rho$ it then
follows that $u\Rightarrow u'$ for some $u'$ with
$\var(u')\subseteq \var(t')$. In case $u \Rightarrow\mv{\tau} u'$
we are done, so assume $u'=u$. Let $\sigma$ be the closed substitution
with $\sigma(x)=\nil$ for all $x\in V$. Since $\sigma(t)
\mv{\tau}$ and $t \precsim_{\rm WIF} u$, we have $\sigma(u)\mv{\tau}$,
so $u \mv{\tau} u''$ for some $u''$.
And $\var(u'') \subseteq \var(u) = \var(u') \subseteq \var(t')$.
\end{proof}

\begin{lemma} \label{key2}
Suppose that for some terms $t,u$, closed substitution $\sigma$, action $a$ and $m>0$:

\begin{enumerate}
  \item \label{lem-prov1} $t \simeq_{\rm WIF} u$;

  \item \label{lem-prov2} $m > \depth_w(u)$;

  \item \label{lem-prov3} $\mc{WCT}(\sigma(u)) \subseteq \{a^{m}, a^{2m}\}$; and

  \item \label{lem-prov4} there is a closed term $p'$ such that
  $\sigma(t)\Rightarrow\mv{\tau} p'$ and $\mc{WCT}(p') = \{a^{2m}\}$.

\end{enumerate}
Then there is a closed term $q'$ such that $\sigma(u)
\Rightarrow\mv{\tau} q'$ and $\mc{WCT}(q') = \{a^{2m}\}$.
\end{lemma}

\begin{proof}
We note that if $p\precsim_{\rm WIF} q$, then $\mc{WCT}(p) \subseteq \mc{WCT}(q)$.
Namely, a process has a weak completed trace $a_1 \cdots a_k$ iff it has a
weak impossible future $(a_1 \cdots a_k,A)$.

According to proviso (\ref{lem-prov4}) of the lemma, two cases can be distinguished:
the trace $\sigma(t)\Rightarrow\mv{\tau} p'$ visits a variable in $t$ or not.

\begin{itemize}
  \item $t\Rightarrow\mv{\tau} t'$ for some $t'$ with $\sigma(t')=p'$.
  Since $\depth_w(t')\leq\depth_w(t)=\depth_w(u)<m$ and $\mc{WCT}(p')=\{a^{2m}\}$, for any \mbox{$y\mathbin\in \var(t')$}
  either $\sigma(y)=\nil$ or $\mc{WCT}(\sigma(y))=\{a^{2m}\}$.
  Since $t\simeq_{\rm WIF} u$, by Lem.~\ref{newvariable},
  $u\Rightarrow\mv{\tau} u'$ for some $u'$ with $\var(u')\subseteq \var(t')$. Hence,
  for any $y\mathbin\in \var(u')$ either $\sigma(y)=\nil$
  or $\mc{WCT}(\sigma(y))=\{a^{2m}\}$.  Since $\depth_w(u')\leq\depth_w(u)<m$,
  it follows that $a^m\mathbin{\notin} \mc{WCT}(\sigma(u'))$.
  Since $\mc{WCT}(\sigma(u))\subseteq\{a^{m}, a^{2m}\}$, we conclude that
  $\mc{WCT}(\sigma(u'))= \{a^{2m}\}$. And $u\Rightarrow\mv{\tau} u'$
  implies $\sigma(u)\Rightarrow\mv{\tau} \sigma(u')$.

  \medskip

  \item $t \Rightarrow t'+x$ for some $t'$ and $x$ with $\sigma(x)\Rightarrow\mv{\tau} p'$.
  Consider the closed substitution $\rho$ defined by $\rho(x)=a^m\nil$
  and $\rho(y)=\nil$ for any $y\neq x$. Then $a^m \in \mc{WCT}(\rho(t))
  = \mc{WCT}(\rho(u))$. Since $\depth_w(u)<m$, in view of the definition of $\rho$
  this implies that $u \Rightarrow u'+x$ for some $u'$. Hence $\sigma(u)\Rightarrow\mv{\tau} p'$.
\qedhere
\end{itemize}
\end{proof}

\begin{proposition} \label{key2cont}
Suppose that for a sound axiomatization $E$ for BCCS($A$) modulo $\simeq_{\rm WIF}$,
closed terms $p,q$, closed substitution $\sigma$, action $a$ and $m>0$:
\begin{enumerate}
  \item \label{prop-prov1} $E\vdash  p \approx q$ or $E\vdash q \approx p$;

  \item \label{prop-prov2} $m > \max \{\depth_w(u)\mid t \approx u \in  E\}$;

  \item \label{prop-prov3} $\mc{WCT}(q) \subseteq \{a^{m},a^{2m}\}$; and

  \item \label{prop-prov4} there is a closed term $p'$ such that $p\Rightarrow\mv{\tau} p'$ and
  $\mc{WCT}(p') = \{a^{2m}\}$.

\end{enumerate}
Then there is a closed term $q'$ such that $q \Rightarrow\mv{\tau}
q'$ and $\mc{WCT}(q') = \{a^{2m}\}$.
\end{proposition}

\begin{proof}

By induction on a derivation of $E\vdash p \approx q$ or $E\vdash q \approx p$.
We only consider $E\vdash p \approx q$; the proof for $E\vdash q \approx p$ is
symmetrical.

The cases where $E\vdash p \approx q$ is derived by reflexivity (i.e., $q=p$)
or symmetry (i.e., $E\vdash p \approx q$ because $E\vdash q \approx p$) are trivial.
We focus on the other cases.

\begin{itemize}

\item Suppose $E\vdash p\approx q$ because $\sigma(t)=p$ and
$\sigma(u)=q$ for some $t\approx u\in E$ or $u\approx t\in E$ and
closed substitution $\sigma$. The claim then follows by
Lem.~\ref{key2}.

\medskip

\item Suppose $E\vdash p\approx q$ because $E\vdash p\approx r$
and $E\vdash r\approx q$ for some $r$. Since $r\simeq_{\rm WIF}
q$, by proviso (\ref{prop-prov3}),
$\mc{WCT}(r)=\mc{WCT}(q)\subseteq \{a^{m},a^{2m}\}$. Since
$p\Rightarrow\mv{\tau} p'$ and $\mc{WCT}(p') = \{a^{2m}\}$,
by induction there is an $r'$ such that $r\Rightarrow\mv{\tau}
r'$ and $\mc{WCT}(r') = \{a^{2m}\}$. Hence, again by induction
there is a $q'$ such that $q\Rightarrow\mv{\tau} q'$ and
$\mc{WCT}(q') = \{a^{2m}\}$.

\medskip

\item Suppose $E\vdash p\approx q$ because $p=p_1+p_2$ and
$q=q_1+q_2$ with $E\vdash p_1\approx q_1$ and $E\vdash p_2\approx
q_2$. Since $p\Rightarrow\mv{\tau} p'$, we have $p_i\Rightarrow \mv{\tau} p'$ for an $i\in\{1,2\}$.
And $\mc{WCT}(q_i)\subseteq\mc{WCT}(q)\subseteq \{a^{m},a^{2m}\}$.
Hence, by induction there is a $q'$ such that $q_i\Rightarrow\mv{\tau} q'$ and $\mc{WCT}(q')
= \{a^{2m}\}$. And $q_i\Rightarrow\mv{\tau} q'$ implies $q\Rightarrow\mv{\tau} q'$.

\medskip

\item Suppose $E\vdash p \approx q$ because $p=\alpha p_1$ and $q=\alpha q_1$
with $E\vdash p_1\approx q_1$. By proviso (\ref{prop-prov4}), $\alpha=\tau$,
and either $\mc{WCT}(p_1)=\{a^{2m}\}$ or there is a $p'$ such
that $p_1\Rightarrow\mv{\tau} p'$ and $\mc{WCT}(p') = \{a^{2m}\}$.
In the first case, $\mc{WCT}(q_1)=\mc{WCT}(p_1)=\{a^{2m}\}$;
and $q \Rightarrow\mv{\tau} q_1$. In the second case, $p\Rightarrow\mv{\tau} p'$ and
$\mc{WCT}(q_1)=\mc{WCT}(q)\subseteq \{a^{m},a^{2m}\}$. So by induction
there is a $q'$ such that $q_1\Rightarrow\mv{\tau} q'$ and
$\mc{WCT}(q') = \{a^{2m}\}$. And $q_1\Rightarrow\mv{\tau} q'$ implies $q\Rightarrow\mv{\tau} q'$.
\qedhere
\end{itemize}
\end{proof}

\begin{thm} \label{thm:negative-WIF-equivalence}
There is no finite, sound, ground-complete axiomatization for BCCS($A$) modulo $\simeq_{\rm WIF}$.
\end{thm}

\begin{proof}
Let $E$ be a finite axiomatization over BCCS($A$) that is
sound modulo $\simeq_{\rm WIF}$. Pick an $m$ greater than the weak depth
of any term in $E$. Consider the closed equation $\tau a^{2m}\nil+\tau(a^m\nil+a^{2m}\nil) \approx
\tau(a^m\nil+a^{2m}\nil)$, which is sound modulo $\simeq_{\rm WIF}$.
We have $\mc{WCT}(\tau(a^m\nil+a^{2m}\nil))=\{a^m,a^{2m}\}$ and $\tau a^{2m}\nil+\tau(a^m\nil+a^{2m}\nil)\mv{\tau}a^{2m}\nil$.
However, clearly there is no closed term $r$ such that
$\tau(a^m\nil+a^{2m}\nil) \Rightarrow\mv{\tau} r$ and $\mc{WCT}(r)=\{a^{2m}\}$.
So according to Prop.~\ref{key2cont}, this closed equation cannot be derived from $E$.
\end{proof}

\begin{rem}
  To explain why there is no counterpart of Thm.~\ref{thm:negative-WIF-equivalence} for
  $\precsim_{\rm WIF}$, we note that Prop.~\ref{key2cont} does not hold if its first requirement is changed into
  $E\vdash p\preccurlyeq q$. Namely, the proof regarding the congruence rule
  for $\tau.\_$ in Prop.~\ref{key2cont} fails for $\precsim_{\rm WIF}$.

  For example, for each $m\geq 0$,
  \[\tau a^{2m}\nil ~\precsim_{\rm WIF}~ \tau(a^{2m}\nil+a^m\nil)\]
  These relations satisfy the third and fourth requirement of Prop.~\ref{key2cont}.
  However, it is not the case that $\tau(a^{2m}\nil+a^m\nil)\Rightarrow\mv{\tau}
  q'$ with $\mc{WCT}(q') = \{a^{2m}\}$.

  We note that these relations can all be derived by means of IF1$'$:
    \[\begin{array}{r} \tau a^{2m}\nil = \tau(a^m(a^m+\nil)) \preccurlyeq \tau(a^{m-1}(a^{m+1}\nil+a\nil)) \vspace{2mm} \\
      \preccurlyeq \tau(a^{m-2}(a^{m+2}\nil+a^2\nil)) \preccurlyeq \cdots \preccurlyeq \tau(a^{2m}\nil+a^m\nil)
  \end{array}\]
\end{rem}

\subsection{Concrete impossible futures equivalence}

There also does not exist a finite, sound, ground-complete axiomatization for
BCCSP($A$) modulo $\simeq_{\rm IF}$. The cornerstone for
this negative result is the following infinite family of closed equations. For each $m\geq 0$,
\[ a^{2m+1}\nil+a(a^m\nil+a^{2m}\nil) ~\approx~ a(a^m\nil+a^{2m}\nil) \]
is sound modulo $\simeq_{\rm IF}$. The proof of the following theorem, which can be found in \cite{CF08},
is very similar to the proof of Thm.~\ref{thm:negative-WIF-equivalence}, and is therefore omitted here.

\begin{thm} \label{thm:negative-IF-equivalence}
There is no finite, sound, ground-complete axiomatization for BCCSP($A$) modulo $\simeq_{\rm IF}$.
\end{thm}

We note that this negative result for $\simeq_{\rm IF}$ is not a direct consequence of
the corresponding negative result for $\simeq_{\rm WIF}$ (Thm.~\ref{thm:negative-WIF-equivalence})
together with Thm.~\ref{thm:correct2}.  Firstly, Thm.~\ref{thm:correct2} disregards soundness,
meaning that a finite, sound, ground-complete axiomatization $E$ for BCCSP($A$) modulo $\simeq_{\rm IF}$
could in theory generate an unsound axiomatization $\mc{A}(E)$ for BCCS($A$) modulo $\simeq_{\rm WIF}$.
Secondly, Thm.~\ref{thm:correct2} only applies to safe axiomatizations $E$.

Interestingly, Thm.~\ref{thm:negative-IF-equivalence} would be a direct consequence of
Thm.~\ref{thm:negative-WIF-equivalence} together with the alternative method described
in Remark \ref{rem:alternative}. Namely, that method does imply soundness of the generated axiomatization,
and is not restricted to safe axiomatizations.

The infinite families of equations that are used to prove Thm.~\ref{thm:negative-WIF-equivalence}
and \ref{thm:negative-IF-equivalence} are also sound modulo weak and concrete 2-nested simulation equivalence, respectively.
Therefore these negative results apply to all BCCS- and BCCSP-congruences that are at least
as fine as impossible futures equivalence and at least as coarse as 2-nested simulation equivalence.
This infers some results from \cite{AFGI04}, where among others concrete 2-nested simulation
and possible futures equivalence were considered.


\section{\texorpdfstring{$\omega$}{Omega}-Completeness for the Impossible Futures Preorder} \label{sec:omega}

This section deals with the existence of $\omega$-complete axiomatizations
in the context of the impossible futures preorder.

\subsection{Inverted substitutions}

Groote \cite{Gro90} introduced the technique of inverted
substitutions to prove that an \emph{equational} axiomatization is
$\omega$-complete. Here we adapt his technique to make it suitable
for \emph{inequational} axiomatizations.
Let $\mathrm{T}(\Sigma)$ and $\mathbb{T}(\Sigma)$ denote the sets of closed and open terms,
respectively, over some signature $\Sigma$.

\begin{thm} \label{thm:inverted-substitutions}
Consider an inequational axiomatization $E$ over $\Sigma$. Suppose
that for each inequation $t\preccurlyeq u$ of which all closed
instances can be derived from $E$, there are a mapping
$R:\mathrm{T}(\Sigma)\rightarrow \mathbb{T}(\Sigma)$ and a closed
substitution $\rho$ such that:
\begin{enumerate}
\item \label{inverted1} $E\vdash t\preccurlyeq R(\rho(t))$ and $E\vdash
R(\rho(u))\preccurlyeq u$; \vspace{2mm}

\item \label{inverted2} $E\vdash R(\sigma(v))\preccurlyeq R(\sigma(w))$ for
each $v\preccurlyeq w\in E$ and closed substitution $\sigma$; and
\vspace{2mm}

\item \label{inverted3} for each function symbol $f$ (with arity $n$) in the
signature, and all closed terms $p_1,\ldots,p_n,q_1,\ldots,q_n$:
\[
\begin{array}{l}
E~\cup~\{p_i\preccurlyeq q_i,~R(p_i)\preccurlyeq R(q_i)\mid i=1,\ldots,n\}~\vdash \vspace{1mm} \\
\hspace*{6cm} R(f(p_1,\ldots,p_n))\preccurlyeq
R(f(q_1,\ldots,q_n))
\end{array}
\]
\end{enumerate}
Then $E$ is $\omega$-complete.
\end{thm}

The underlying idea of Thm.~\ref{thm:inverted-substitutions} is that each derivation of a closed inequation $p\preccurlyeq q$ from $E$
induces, by requirements (\ref{inverted2}) and (\ref{inverted3}), a derivation of the open
inequation $R(p)\preccurlyeq R(q)$ from $E$. Since $E\vdash\rho(t)\preccurlyeq\rho(u)$,
it follows using requirement (\ref{inverted1}) that $E\vdash t\preccurlyeq R(\rho(t))\preccurlyeq R(\rho(u))\preccurlyeq u$.
We now proceed to prove this theorem.

\begin{proof}
Let $t, u$ be terms such that for each closed substitution
$\sigma$,
  \[  \sigma(t) ~\preccurlyeq~ \sigma(u)\]
By assumption, there are a mapping $R:\mathrm{T}(\Sigma)\rightarrow \mathbb{T}(\Sigma)$
and a closed substitution $\rho$ such that the three requirements of
Thm.~\ref{thm:inverted-substitutions} are satisfied. We have to prove that
$E\vdash t\preccurlyeq u$. This is an immediate corollary of the
following claim, for all closed terms $p,q$:
\[ E~\vdash~ p~\preccurlyeq~ q ~~\implies~~ E~\vdash~ R(p)~\preccurlyeq~ R(q)\]
Namely, by assumption, $E\vdash \rho(t) \preccurlyeq \rho(u)$, and
then the claim above implies that $E\vdash R(\rho(t))\preccurlyeq
R(\rho(u))$. So by requirement (\ref{inverted1}) of Thm.~\ref{thm:inverted-substitutions}, $E\vdash t \preccurlyeq u$.

We prove the claim by induction on a derivation of $E\vdash p\preccurlyeq q$. We have to check the four kinds of inference rules.
\begin{itemize}
    \item $p = q$. Then $R(p)=R(q)$.

    \medskip

    \item $p \preccurlyeq q$ is an instance of some $v\preccurlyeq w\in E$ and a closed substitution
    $\sigma$. By requirement (\ref{inverted2}) of Thm.~\ref{thm:inverted-substitutions}, $E\vdash R(p)\preccurlyeq R(q)$.

    \medskip

    \item $E\vdash p \preccurlyeq q$ has been proved by $E\vdash p \preccurlyeq r$ and $E\vdash r \preccurlyeq
    q$, for some $r$. By induction, $E\vdash R(p) \preccurlyeq R(r)$ and
    $E\vdash R(r) \preccurlyeq R(q)$. So $E\vdash R(p) \preccurlyeq R(q)$.

    \medskip

    \item $p=f(p_1,\ldots,p_n)$ and $q=f(q_1,\ldots,q_n)$,
    and $E\vdash p\preccurlyeq q$ has been proved by $E\vdash p_i\preccurlyeq
    q_i$ for $i=1,\ldots,n$. By induction,
    $E\vdash R(p_i)\preccurlyeq R(q_i)$ for $i=1,\ldots,n$. So by requirement (\ref{inverted3}) of Thm.~\ref{thm:inverted-substitutions},
    $E \vdash R(f(p_1,\ldots,p_n))\preccurlyeq R(f(q_1,\ldots,q_n))$.
    \qedhere
  \end{itemize}
\end{proof}

\subsection{Infinite alphabet}

We show that the axiomatization consisting of
A1-4+IF1-2 is $\omega$-complete, provided the alphabet is
infinite. The proof is based on inverted substitutions.

\begin{thm} \label{yes}
If $|A|=\infty$, then \emph{A1-4+IF1-2} is $\omega$-complete for BCCSP($A$).
\end{thm}

\proof
We define the closed substitution $\rho$ by $\rho(y)=a_y\nil$, where $a_y$ is a
distinct action in $A$ for each $y\in V$ that occurs in neither $t$ nor $u$.
Such actions exist because $A$ is infinite. We define the mapping
$R$ from closed to open BCCSP($A$) terms as follows:
\[\left\{
\begin{tabular}{lll}
    $R(\nil)$   & $=\nil$\\
    $R(at)$   & $=y$ & if $a=a_y$ for some $y\in V$\\
    $R(at)$   & $=aR(t)$ & if $a\neq a_y$ for all $y\in V$\\
    $R(t+u)$   & $=R(t)+R(u)$
\end{tabular}
\right.
\]

Consider a pair of BCCSP($A$) terms $t,u$ such that all closed
instances of $t\preccurlyeq u$ can be derived from A1-4+IF1-2.
We check the three requirements of Thm.~\ref{thm:inverted-substitutions}.
\begin{enumerate}

\item[(\ref{inverted1})] Since $t$ and $u$ do not contain actions of the form
$a_y$\hspace{-1pt}, clearly $R(\rho(t))\mathbin=t$ and $R(\rho(u))\mathbin=u$.

\medskip

\item[(\ref{inverted2})] For A1-4, the proof is trivial, because for each of these four axioms $v\approx w$,
$R(\sigma(v))\approx R(\sigma(w))$ is always a substitution instance of the axiom itself. We check the remaining
cases IF1 and IF2. Let $\sigma$ be a closed substitution. With regard to IF1 we distinguish two cases.

\begin{itemize}[label=$-$]

\item $a = a_y$ for some $y\in V$. Then \[R(a_y(\sigma(x_1)+
\sigma(x_2))) = y = y+y=
R(a_y(\sigma(x_1))+a_y(\sigma(x_2))).\]

\medskip

\item $a\neq a_y$ for all $y\in V$. Then using IF1 we derive
\[\eqalign{
  R(a(\sigma(x_1)+\sigma(x_2))) &= a(R(\sigma(x_1))+R(\sigma(x_2)))\cr
&\preccurlyeq aR(\sigma(x_1))+aR(\sigma(x_2)) \cr
&= R(a\sigma(x_1)+a\sigma(x_2)).}
\]

\end{itemize}

\medskip

\noindent
With regard to IF2 we distinguish two cases as well.

\begin{itemize}[label=$-$]

\item $a=a_y$ for some $y\in V$. Then 
\[\eqalign{R(a_y\sigma(x_1) +a_y(\sigma(x_2)+\sigma(x_3))) 
&= y+y \cr
&= y+y+y \cr
&= R(a_y(\sigma(x_1)+
\sigma(x_2)) + a_y\sigma(x_1) +a_y(\sigma(x_2)+ \sigma(x_3))).}
\]

\medskip

\item $a\neq a_y$ for all $y\in V$. Then using IF2 we derive
\[\eqalign{R(a\sigma(x_1)+a(\sigma(x_2)+\sigma(x_3))) 
&= aR(\sigma(x_1))+a(R(\sigma(x_2))+R(\sigma(x_3)))\cr
&\approx a(R(\sigma(x_1))+R(\sigma(x_2)))\cr
&\phantom{\approx{}}+aR(\sigma(x_1))+a(R(\sigma(x_2))+R(\sigma(x_3))) \cr
&=R(a(\sigma(x_1)+\sigma(x_2))+a\sigma(x_1)+a(\sigma(x_2)+\sigma(x_3))).}
\]
\end{itemize}

\medskip

\item[(\ref{inverted3})] Consider the operator $\_+\_$. From $R(p_1)\preccurlyeq
R(q_1)$ and $R(p_2)\preccurlyeq R(q_2)$ we derive
\[R(p_1+p_2)=R(p_1)+R(p_2)\preccurlyeq R(q_1)+R(q_2)=R(q_1+q_2).\]

\medskip

\noindent
Consider the prefix operator $a\_$. We distinguish two cases.
\begin{itemize}[label=$-$]
\item $a=a_y$ for some $y\in V$. Then $R(a_yp_1)=y=R(a_yq_1)$.
\item $a\neq a_y$ for all $y\in V$. Then from $R(p_1)\preccurlyeq
R(q_1)$ we derive \[R(ap_1)=aR(p_1)\preccurlyeq aR(q_1)=R(aq_1).\eqno{\qEd}\]
\end{itemize}
\end{enumerate}

\begin{cor}
If $|A|=\infty$, then \emph{A1-4+WIF1+WIF2$'$+W1} is $\omega$-complete for BCCS($A$).
\end{cor}

This result was already obtained in \cite{VM01} (they do not refer to
$\omega$-completeness explicitly, but their completeness proof works for
open terms, see \cite[Thm.~5]{VM01}).
Here it is a direct corollary of the link established in Sect.~\ref{sec:link},
by the same reasoning as in Sect.~\ref{sec:wifp}.

\newpage
\subsection{Finite alphabet}
\label{sec:finite-alphabet}

\subsubsection{$1<|A|<\infty$.}
We prove that the inequational theory of BCCS($A$) modulo
$\precsim_{\rm WIF}$ does not have a finite basis in case of
a finite alphabet $A$ with at least two elements.
Pick an $a\in A$ and an $x\in V$. The cornerstone for this negative result
is the following infinite family of inequations, for each $m\geq 0$:
\[\tau(a^mx) +  \Phi_m ~\preccurlyeq~ \Phi_m \]
with
\[\Phi_m ~=~ \tau(a^mx+x) + \sum_{b\in A} \tau(a^mx+a^mb\nil)\]
We argue that these inequations are sound modulo $\precsim_{\rm WIF}$.
For any closed substitution $\rho$ we have $\mc{WT}(\rho(\tau(a^mx))) \subseteq \mc{WT}(\rho(\Phi_m))$
and $\rho(\Phi_m)\mv{\tau}$. To argue that $\rho(\tau(a^mx) +
\Phi_m)$ and $\rho(\Phi_m)$ have the same weak impossible futures,
it suffices to consider the transition $\rho(\tau(a^mx)+\Phi_m)\mv{\tau}a^m\rho(x)$.
If $\rho(x)=\nil$, then $\rho(\Phi_m)\mv{\tau}a^m\nil + \nil$ generates the same
weak impossible futures $(\varepsilon,B)$.  If, on the other hand,
$b\mathbin\in\mc{I}(\rho(x))$ for some $b\mathbin\in A$, then
$\mc{WT}(a^m\rho(x)+a^mb\nil)=\mc{WT}(a^m\rho(x))$, so
$\rho(\Phi_m)\mv{\tau} a^m\rho(x)+a^mb\nil$ generates the same
weak impossible futures $(\varepsilon,B)$.

We extend the notions of weak traces from closed to open terms, allowing weak traces
of the form $a_1 \cdots a_k x \in A^\ast V$.  This is done by treating each variable occurrence $x$
in a term as if it were a subterm $x\nil$ with $x$ a concrete action.
For example, under this convention the weak completed traces of $\Phi_m$ are $a^mx$, $x$ and $a^mb$ for all $b\in A$.
We write $\mc{WT}_V(t)$ for the set of weak traces of $t$ that end with a variable.
So $\mc{WT}_V(\Phi_m) =\{a^{m}x, x\}$.
Note that $\mc{WT}_V(t) \mathbin\subseteq \mc{WT}_V(u)$ implies $\mc{WT}_V(\sigma(t))
\mathbin\subseteq \mc{WT}_V(\sigma(u))$ for any terms $u,v$ and substitution~$\sigma$.

\begin{rem}\label{traces-substitution}
Let $m>\depth_w(t)$ and $\rho$ a closed substitution.
Then $a_1\cdots a_m \in \mc{WT}(\rho(t))$ iff there are a $k < m$ and $y\mathbin\in V$
such that $a_1\cdots a_k y \in \mc{WT}_V(t)$ and $a_{k+1} \cdots a_m \in \mc{WT}(\rho(y))$.
\end{rem}

\begin{lemma}\label{trace-preservation}
Let $|A|>1$. If $t \precsim_{\rm WIF} u$, then $\mc{WT}(t)= \mc{WT}(u)$.
\end{lemma}

\begin{proof}
It is not hard to see (by substituting $\nil$ for all variables in $t$ and $u$)
that $t \precsim_{\rm WIF} u$ implies that $t$ and $u$ must have the same weak traces ending with an action.

Pick distinct actions $a,b\mathbin\in A$
and an injection $\iota: V \rightarrow \mathbb{Z}_{>0}$. Let
$m=\depth_w(u)+1$. Define the closed substitution $\rho$ by
$\rho(z)=a^{\iota(z){\cdot}m}b\nil$ for all $z\in V$.
Since $t\precsim_{\rm WIF} u$, we have $\mc{WT}(\rho(t))= \mc{WT}(\rho(u))$.
As $m>\depth_w(t)=\depth_w(u)$, $\iota$ is an injection, and $a$ and $b$ are distinct,
it is not hard to see, by Remark~\ref{traces-substitution}, that this implies $\mc{WT}_V(t) \mathbin= \mc{WT}_V(u)$.
\end{proof}

\begin{lemma} \label{newvariable3}
Let $|A|>1$. If $t\precsim_{\rm WIF} u$ and $t\Rightarrow\mv{\tau} t'$,
then there is a term $u'$ with $u\Rightarrow\mv{\tau} u'$ and $\mc{WT}_V(u')\subseteq\mc{WT}_V(t')$.
\end{lemma}

\begin{proof}
Define the closed substitution $\rho$ as in the previous proof.
Since $t\precsim_{\rm WIF} u$ and $\rho(t)\Rightarrow\rho(t')$, there must
be a closed term $q$ with $\rho(u) \Rightarrow q$ and $\mc{WT}(q)\subseteq
\mc{WT}(\rho(t'))$. Since $\rho(z)$ is a BCCSP($A$) term for all $z\in V$,
it follows that there is a term $u'$ with $u \Rightarrow u'$ and $\rho(u')=q$.
Again, as $m>\depth_w(u)$, $\iota$ is an injection, and $a$ and $b$ are distinct,
it is not hard to see that $\mc{WT}(\rho(u'))\subseteq
\mc{WT}(\rho(t'))$ implies $\mc{WT}_V(u')\subseteq\mc{WT}_V(t')$.
In case $u \Rightarrow\mv{\tau} u'$ we are done, so assume $u'=u$.
Since $t \precsim_{\rm WIF} u$ and $t \mv{\tau}$, clearly $u \mv{\tau} u''$ for some $u''$.
And $\mc{WT}_V(u'') \subseteq \mc{WT}_V(u) = \mc{WT}_V(u') \subseteq \mc{WT}_V(t')$.
\end{proof}

\begin{lemma} \label{no-traces1}
  Let $1<|A|<\infty$. Suppose that for some terms $t,u$, substitution $\sigma$, action $a$ and $m>0$:
  \begin{enumerate}
  \item \label{traces-prov1} $t \precsim_{\rm WIF} u$;

  \item \label{traces-prov2} $m \geq \depth_w(u)$; and

  \item \label{traces-prov4} $\sigma(t)\Rightarrow\mv{\tau} \hat t$ for a term $\hat t$
  without weak traces $x$ and $a^mb$ for any $b\in A$.
 \end{enumerate}
  Then $\sigma(u)\Rightarrow\mv{\tau}\hat u$ for a term $\hat u$ without weak traces $x$ and $a^mb$ for any $b\in A$.
\end{lemma}
\begin{proof}
Based on proviso (\ref{traces-prov4}) there are two cases to consider.
\begin{itemize}
\item $t \Rightarrow\mv{\tau} t'$ for some term $t'$ with $\sigma(t') = \hat t$.
  By Lem.~\ref{newvariable3} there is a term $u'$ with
  $u\Rightarrow\mv{\tau} u'$ and $\mc{WT}_V(u')\subseteq \mc{WT}_V(t')$.
  So $\sigma(u)\Rightarrow\mv{\tau} \sigma(u')$. By proviso (\ref{traces-prov4}) of the lemma,
  $x\notin \mc{WT}_V(\sigma(t')) \supseteq \mc{WT}_V(\sigma(u'))$. Assume, toward a contradiction,
  that $a^mb\in\mc{WT}(\sigma(u'))$ for some $b\in A$. Since $m \geq \depth_w(u)$,
  clearly there is a $k \leq m$ and a $y\mathbin\in V$ such that $a^k y \in \mc{WT}_V(u')$ and
  $a^{m-k}b \in \mc{WT}(\sigma(y))$. Since $a^k y\in\mc{WT}_V(u')\subseteq \mc{WT}_V(t')$,
  it follows that $a^mb \in \mc{WT}(\sigma(t'))$, which contradicts proviso (\ref{traces-prov4}) of the lemma.
  Hence we can take $\hat u=\sigma(u')$.
\item $y \in \mc{WT}_V(t)$ and $\sigma(y)\Rightarrow\mv{\tau} \hat t$ for some $y\mathbin\in V$.
  Lem.~\ref{trace-preservation} yields $y \mathbin\in \mc{WT}_V(u)$, so $\sigma(u)\Rightarrow\mv{\tau} \hat t$.
  Hence we can take $\hat u=\hat t$.\qedhere
\end{itemize}
\end{proof}

\begin{proposition} \label{no-traces2}
  Let $1<|A|<\infty$, and let the axiomatization $E$ be sound for BCCS($A$) modulo $\precsim_{\rm WIF}$.
  Suppose that for some terms $v,w$, action $a$ and $m>0$:
  \begin{enumerate}
  \item \label{notraces-prov1} $E\vdash  v \preccurlyeq w$;

  \item \label{notraces-prov2} $m \geq \max \{\depth_w(u) \mid t \preccurlyeq u \in  E\}$; and

  \item \label{notraces-prov4} $v\Rightarrow\mv{\tau} \hat v$ for a term $\hat v$
   without weak traces $x$ and $a^mb$ for any $b\in A$.
 \end{enumerate}
  Then $w\Rightarrow\mv{\tau}  \hat w$ for a term $\hat w$ without weak traces $x$ and $a^mb$ for any $b\in A$.
\end{proposition}

\begin{proof}
By induction on a derivation of $E\vdash v \preccurlyeq w$. Reflexivity is
trivial; we focus on the other cases.
\begin{itemize}

\item Suppose $E\vdash v\preccurlyeq w$ because $\sigma(t)=v$ and
$\sigma(u)=w$ for some $t\preccurlyeq u\in E$ and substitution
$\sigma$. The claim then follows by Lem.~\ref{no-traces1}.

\medskip

\item Suppose $E\vdash v\preccurlyeq w$ because $E\vdash
v\preccurlyeq u$ and $E\vdash u\preccurlyeq w$ for some $u$.
By induction, $u \Rightarrow\mv{\tau} \hat u$ for a term $\hat u$
without weak traces $x$ and $a^mb$ for any $b\in A$.. Hence, again by induction,
$w\Rightarrow\mv{\tau} \hat w$ for a term $\hat w$ without weak traces $x$ and $a^mb$ for any $b\in A$.

\medskip

\item Suppose $E\vdash v\preccurlyeq w$ because $v=v_1+v_2$ and
$w=w_1+w_2$ with $E\vdash v_1 \preccurlyeq w_1$ and $E\vdash
v_2\preccurlyeq w_2$. Since $v \Rightarrow\mv{\tau} \hat v$,
we have $v_1 \Rightarrow\mv{\tau} \hat v$ or $v_2
\Rightarrow\mv{\tau} \hat v$. Assume, without loss of generality,
that $v_1 \Rightarrow\mv{\tau} \hat v$.  By induction, $w_1 \Rightarrow\mv{\tau} \hat w$
for a term $\hat w$ without weak traces $x$ and $a^mb$ for any $b\in A$.
And $w \Rightarrow\mv{\tau} \hat w$.

\medskip

\item Suppose $E\vdash v \preccurlyeq w$ because $v=\alpha v_1$ and
$w=\alpha w_1$ with $E\vdash v_1\preccurlyeq w_1$.
By proviso (\ref{notraces-prov4}) of the proposition, $\alpha=\tau$,
and either $v_1=\hat v$ or $v_1 \Rightarrow\mv{\tau} \hat v$. In the first case,
by proviso (\ref{notraces-prov4}) of the proposition and Lem.~\ref{trace-preservation},
$w_1$ has no  weak traces $x$ and $a^mb$ for any $b\in A$. And $w \mv{\tau} w_1$.
In the second case, by induction, $w_1\Rightarrow\mv{\tau} \hat w$
for a term $\hat w$ without weak traces $x$ and $a^mb$ for any $b\in A$. And $w\Rightarrow\mv{\tau} \hat w$.
\qedhere
\end{itemize}
\end{proof}

\begin{thm} \label{thm:alphabetn}
If $1<|A|<\infty$, then the inequational theory of BCCS($A$) modulo $\precsim_{\rm WIF}$ does not have a finite basis.
\end{thm}

\begin{proof}
Let $E$ be a finite axiomatization over BCCS($A$) that is
sound modulo $\precsim_{\rm WIF}$. Let $m$ be greater than the
weak depth of any term in $E$.  According to Prop.~\ref{no-traces2}, the
inequation $\tau(a^mx) + \Phi_m \preccurlyeq \Phi_m$ cannot be
derived from $E$.  Yet it is sound modulo $\precsim_{\rm WIF}$.
\end{proof}

Likewise, the inequational theory of BCCSP($A$) modulo $\precsim_{\rm IF}$ also does not have
a finite basis in case of a non-singleton finite alphabet. If $A$ has
two distinct elements $a$ and $b$, then the cornerstone for this negative result
is the following infinite family of inequations, for each $m\geq 0$:
\[a^{m+1}x +  \Psi_m ~\preccurlyeq~ \Psi_m \]
with
\[\Psi_m ~=~ a(a^mx+x) + \sum_{b\in A} a(a^mx+a^mb\nil)\]
The proof of the following theorem, which can be found in \cite{CF08},
is very similar to the proof of Thm.~\ref{thm:alphabetn}, and is therefore omitted here.

\begin{thm}
If $1 < |A|<\infty$, then the inequational theory of BCCSP($A$) modulo $\precsim_{\rm IF}$ does not have a finite basis.
\end{thm}

\subsubsection{$|A|=1$.}
We prove that the inequational theory of
BCCS($A$) modulo $\precsim_{\rm WIF}$ does not have
a finite basis in case of a singleton alphabet. The cornerstone
for this negative result is the following infinite family of
inequations, for each $m\geq 0$:
\[ a^m x ~\preccurlyeq~  a^m x + x\]
If $|A|=1$, then these inequations are clearly sound modulo
$\precsim_{\rm WIF}$. In particular, for any closed substitution $\rho$,
$\mc{WT}(\rho(x))\subseteq \mc{WT}(\rho(a^m x))$.
\pagebreak[3]

For $m\geq 1$, the inequations above show that Lem.~\ref{trace-preservation}
fails if $|A|=1$. We now formulate a weaker variant of this lemma for $|A|=1$.

\begin{lemma}\label{trace-preservation2}
Let $A=\{a\}$. If $t \precsim_{\rm WIF} u$, then $\mc{WT}_V(t) \subseteq \mc{WT}_V(u)$.
\end{lemma}

\begin{proof}
Select an injection $\iota: V \rightarrow \mathbb{Z}_{>0}$.
Let $m=\depth_w(u)+1$. Define the closed substitution $\rho$ by
$\rho(z)=a^{\iota(z){\cdot}m}\nil$ for all $z\in V$. Since $t \precsim_{\rm WIF} u$,
we have $\mc{WT}(\rho(t))=\mc{WT}(\rho(u))$. Consider any
$a^ky \in \mc{WT}_V(t)$. Then $a^{k+\iota(y){\cdot}m} \in \mc{WT}(\rho(t))=\mc{WT}(\rho(u))$.
As $m>\depth_w(u)\geq k$ and $\iota$ is an injection,
it is not hard to see that this implies $a^ky \in \mc{WT}_V(u)$.
\end{proof}
 
\begin{lemma} \label{1}
 Let $A=\{a\}$. Suppose that for some terms $t,u$, substitution $\sigma$, variable $x$ and $m>0$:
  \begin{enumerate}
  \item $t \precsim_{\rm WIF} u$;
  \item $m>\depth_w(u)$; and
  \item $x\in\mc{WT}_V(\sigma(u))$ and $a^kx\notin\mc{WT}_V(\sigma(u))$ for all $1\leq k < m$.
 \end{enumerate}
Then $x\in\mc{WT}_V(\sigma(t))$ and $a^kx\notin\mc{WT}_V(\sigma(t))$ for all $1\leq k < m$.
\end{lemma}

\begin{proof}
Since $x\in \mc{WT}_V(\sigma(u))$, clearly there is a $y\in
\mc{WT}_V(u)$ with $x \in \mc{WT}_V(\sigma(y))$. Consider the closed
substitution $\rho$ defined by $\rho(y)=a^m\nil$ and $\rho(z)=\nil$ for all $z\neq y$.
Since $y\in \mc{WT}_V(u)$, we have $a^m \in \mc{WT}(\rho(u))$.
So $t \precsim_{\rm WIF} u$ implies $a^m\in\mc{WT}(\rho(t))$.
Since $m>\depth_w(u)=\depth_w(t)$, clearly
there are some $\ell < m$ and $z\in V$ such that $a^\ell z \in \mc{WT}_V(t)$
and $a^{m-\ell} \in \mc{WT}(\rho(z))$. As $\ell<m$, it follows from the definition of $\rho$
that $z=y$. Since $a^\ell y\in \mc{WT}_V(t)$ and $x\in \mc{WT}_V(\sigma(y))$,
it follows that $a^\ell x \in \mc{WT}_V(\sigma(t))$. By Lem.~\ref{trace-preservation2}, $a^k x
\not\in \mc{WT}_V(\sigma(t))$ for all $1\leq k < m$. Hence we obtain $\ell=0$.
\end{proof}

\begin{proposition} \label{lemmaalphabet1}
Let $A=\{a\}$, and let the axiomatization $E$ be sound for BCCS($A$) modulo $\precsim_{\rm WIF}$.
Suppose that for some terms $v,w$, variable $x$ and $m>0$:
 \begin{enumerate}
  \item \label{alphabet-prov1} $E\vdash  v \preccurlyeq w$;

  \item \label{alphabet-prov2} $m > \max \{\depth_w(u) \mid t \preccurlyeq u \in  E\}$; and

  \item \label{alphabet-prov3} $x \in \mc{WT}_V(w)$ and
  $a^k x \not\in \mc{WT}_V(w)$ for all $1\leq k < m$.
 \end{enumerate} Then $x \in \mc{WT}_V(v)$ and $a^k x \not\in \mc{WT}_V(v)$ for all $1\leq k < m$.
\end{proposition}

\begin{proof}
By induction on a derivation of $E\vdash v \preccurlyeq w$. Reflexivity
is trivial; we focus on the other cases.
\begin{itemize}

\item Suppose $E\vdash v\preccurlyeq w$ because $\sigma(t)=v$ and
$\sigma(u)=w$ for some $t\preccurlyeq u\in E$ and substitution
$\sigma$. The claim then follows by Lem.~\ref{1}.

\medskip

\item Suppose $E\vdash v\preccurlyeq w$ because $E\vdash
v\preccurlyeq u$ and $E\vdash u\preccurlyeq w$ for some $u$. By
induction, $x \in \mc{WT}_V(u)$ and $a^k x \not\in \mc{WT}_V(u)$ for
all $1\leq k < m$. Hence, again by induction, $x \in \mc{WT}_V(v)$ and
$a^k x \not\in \mc{WT}_V(v)$ for all $1\leq k < m$.

\medskip

\item Suppose $E\vdash v\preccurlyeq w$ because $v=v_1+v_2$ and
$w=w_1+w_2$ with $E\vdash v_1 \preccurlyeq w_1$ and $E\vdash
v_2\preccurlyeq w_2$. Since $x \in \mc{WT}_V(w)$, we have $x \in
\mc{WT}_V(w_1)$ or $x \in \mc{WT}_V(w_2)$.  Assume, without loss of
generality, that $x \in \mc{WT}_V(w_1)$.  Since $a^k x \not\in
\mc{WT}_V(w)$ for all $1\leq k < m$, surely $a^k x \not\in
\mc{WT}_V(w_1)$ for all $1\leq k < m$.  By induction, $x \in
\mc{WT}_V(v_1)$, and hence $x \in \mc{WT}_V(v)$.  For all $1\leq k < m$
we have $a^k x \not\in \mc{WT}_V(w)$, and hence $a^k x \not\in
\mc{WT}_V(v)$, by Lem.~\ref{trace-preservation2}.

\medskip

\item Suppose $E\vdash v \preccurlyeq w$ because $v=\alpha v_1$ and
$w=\alpha w_1$ with $E\vdash v_1\preccurlyeq w_1$.
By proviso (\ref{alphabet-prov3}) of the proposition, $\alpha=\tau$,
$x \in \mc{WT}_V(w_1)$ and $a^k x \not\in \mc{WT}_V(w_1)$ for all $1\leq k < m$.
By induction, $x \in \mc{WT}_V(v_1)$ and $a^k x \not\in \mc{WT}_V(v_1)$ for all
$1\leq k < m$. Hence $x \in \mc{WT}_V(v)$ and $a^k x \not\in \mc{WT}_V(v)$ for
all $1\leq k < m$.
\qedhere
\end{itemize}

\end{proof}

\begin{thm}\label{thm:alphabet1}
If $|A|=1$, then the inequational theory of BCCS($A$) modulo $\precsim_{\rm WIF}$ does not have a finite basis.
\end{thm}

\begin{proof}
Let $E$ be a finite axiomatization over BCCS($A$) that is
sound modulo $\precsim_{\rm WIF}$. Let $m$ be greater than the
weak depth of any term in $E$.  According to Prop.~\ref{lemmaalphabet1},
the inequation $a^mx \preccurlyeq a^m x + x$ cannot be derived
from $E$.  Yet, since $|A|=1$, it is sound modulo $\precsim_{\rm WIF}$.
\end{proof}

Likewise, the inequational theory of BCCSP($A$) modulo $\precsim_{\rm IF}$
does not have a finite basis in case of a singleton alphabet.
This negative result is based on the same infinite family of
inequations as for the weak case: $a^m x \preccurlyeq a^m x + x$ for each $m\geq 0$.
The proof of the following theorem is more or less identical to the
proof of Thm.~\ref{thm:alphabet1}, and is therefore omitted here.

\begin{thm}
If $|A|=1$, then the inequational theory of BCCSP($A$) modulo $\precsim_{\rm IF}$ does not have a finite basis.
\end{thm}

\section{Conclusion} \label{sec:con}

We have introduced a method to transform an axiomatization for a concrete semantics in the context
of the process algebra BCCSP to an axiomatization for a corresponding weak semantics with regard to BCCS,
such that the properties ground-completeness and $\omega$-completeness are preserved.
Trace, (a form of) completed trace, failures and impossible futures semantics are within the realm of this transformation method.

Exploiting this approach, we obtained axiomatizations for the weak trace, completed trace and failures
preorders and equivalences.
Moreover, we performed a comprehensive and systematic study on the axiomatizability of concrete
and weak impossible futures semantics over BCCSP and BCCS\@.
Table~\ref{tab:summary} presents an overview, where $+$ indicates that a finite
axiomatization exists, while $-$ indicates that a finite axiomatization does
not exist.  The table expands in two dimensions:
ground-completeness vs.\ $\omega$-completeness and preorder vs.\ equivalence.
When necessary, we distinguish two categories, according to the cardinality of the alphabet $A$:
finite or infinite.

\begin{table}[htb]
\begin{center}
\begin{tabular}{|c|c|c|c|c|}
\cline{1-5}
 &  \multicolumn{2}{|c}{ground-completeness} & \multicolumn{2}{|c|}{$\omega$-completeness} \\
\cline{2-5} 
 &   \multicolumn{2}{|c|}{$1\leq |A|\leq \infty$}   &  $|A|\mathord=\infty$  &  $1\mathord\leq |A|\mathord<\infty$ \\
\hline
(concrete or weak) imp.\ futures preorder   &    \multicolumn{2}{|c|}{\bf +}   &   {\bf +} &  {\bf --}\\
\hline
\!(concrete or weak) imp.\ futures equivalence\!   &   \multicolumn{2}{|c|}{\bf --}   &   {\bf --} &  {\bf --}\\
\hline
\end{tabular}\caption{\label{tab:summary}Summary of results for impossible futures semantics}
\end{center}
\end{table}

Impossible futures semantics is the first example that, in the context of BCCSP/BCCS,
affords a ground-complete axiomatization modulo the preorder,
while missing a ground-complete axiomatization modulo the equivalence.

\subsection*{Acknowledgement} This research was initiated by questions from Jos Baeten,
on axiomatizing impossible futures semantics, and from David de Frutos-Escrig,
on relating the axiomatizability of concrete and weak impossible futures semantics.

\bibliographystyle{plain}

\end{document}